\documentclass[sigconf]{acmart}

\AtBeginDocument{%
  }

\copyrightyear{2024}
\acmYear{2024}
\setcopyright{rightsretained}
\acmConference[CIKM '24]{Proceedings of the 33rd ACM International Conference on Information and Knowledge Management}{October 21--25, 2024}{Boise, ID, USA}
\acmBooktitle{Proceedings of the 33rd ACM International Conference on Information and Knowledge Management (CIKM '24), October 21--25, 2024, Boise, ID, USA}
\acmDOI{10.1145/3627673.3679780}
\acmISBN{979-8-4007-0436-9/24/10}

\usepackage{macros}

\begin{document}

\title[The Impact of External Sources on the Friedkin--Johnsen Model]{The Impact of External Sources on the Friedkin--Johnsen Model}

\author{Charlotte Out}
\authornote{Equal contribution.}
\affiliation{%
  \institution{University of Cambridge}
  \country{United Kingdom}
  }
\email{ceo33@cam.ac.uk}

\author{Sijing Tu}
\authornotemark[1]
\affiliation{%
  \institution{KTH Royal Institute of Technology}
  \country{Sweden}
}
\email{sijing@kth.se}

\author{Stefan Neumann}
\authornote{Equal senior authorship.}
\affiliation{%
 \institution{TU Wien}
 \country{Austria}
}
\email{stefan.neumann@tuwien.ac.at}

\author{Ahad N. Zehmakan}
\authornotemark[2]
\affiliation{%
  \institution{Australian National University}
  \country{Australia}
}
\email{ahadn.zehmakan@anu.edu.au}

\renewcommand{\shortauthors}{Charlotte Out, Sijing Tu, Stefan Neumann, and Ahad N. Zehmakan}

\begin{abstract}
 To obtain a foundational understanding of timeline algorithms and viral content in shaping public opinions, computer scientists started to study augmented versions of opinion formation models from sociology. 
In this paper, we generalize the popular Friedkin--Johnsen model to include the effects of external media sources on opinion formation. Our goal is to mathematically analyze the influence of biased media, arising from factors such as manipulated news reporting or the phenomenon of false balance.
Within our framework, we examine the scenario of two opposing media sources, which do not adapt their opinions like ordinary nodes, and analyze the conditions and the number of periods required for radicalizing the opinions in the network. %
When both media sources possess equal influence, we theoretically characterize the final opinion configuration. 
In the special case where there is only a single media source present, we 
prove that media sources which do not adapt their opinions are significantly more powerful than those which do. Lastly, we conduct the experiments on real-world and synthetic datasets, showing that our theoretical guarantees closely align with experimental simulations.
\end{abstract}




\maketitle

\section{Introduction}\label{sec:Introduction}

Understanding the impact of social media and of conventional media (such as TV or
newspapers) on modern societies has been an active research topic of the last
decade.  This has led to a large body of empirical work, which obtains real-world data
from various media sources and analyzes it to
obtain insights into societal phenomena \cite{myers2012information}.

To gain an enhanced theoretical understanding, computer scientists have recently started to study opinion formation models from sociology.
They augment these models with abstractions of how online social networks impact
the opinion formation process, for instance, when the social network provider
(like Facebook or X, previously known as Twitter) aims to minimize the
disagreement between the users~\cite{chitra2020analyzing}
or in the context of confirmation bias and friend-of-friend
recommendations~\cite{bhalla2023local}.

One shortcoming of previous works is that they usually analyze graphs based on online social networks
without any external influence. However, we argue that the discussion on online
platforms like X or Reddit is also influenced by external sources, such
as conventional media (like newspapers or TV), which are often consumed
independent of online social networks.

The influence of these external media sources on the opinion formation process can be particularly stark in the event that the media source is \emph{biased}, i.e., it does not truthfully reflect the
average of the opinions. External media sources can be biased, for instance, due to phenomena like \emph{false balance} or \emph{bothsidesism}, where media aim to present
both sides of a conflict but unintentionally overrepresent one side of the
conflict. This received attention when media gave too much attention to doubters of climate
change~\cite{boykoff2004balance}. Bias can also be the result of intentional manipulations. For example, the Polish PiS party, which formed a majority government from $2015$ to $2023$, has been blamed for engineering the media coverage to support their own
agenda~\cite{poland}. As exposure to this biased information can lead to negative societal outcomes, including group polarization, intolerance of dissent, and political segregation~\cite{spohr2017fake}, it is critical to gain an enhanced understanding of the effect of (biased) external media sources on the opinion formation process.

\spara{Our contributions.}
In this paper, we investigate how
much impact external sources can have on the opinion-formation process in a
social network. In particular, we propose an augmented version of the popular
Friedkin--Johnsen (FJ) model~\cite{friedkin1990social}, in which one or two
external sources are added (see \cref{sec:StubbornMediaSourcesSetting} for
details).
We use our model to study the setting in which the external media sources are
\emph{biased}, i.e., they do not report the average of the individuals' opinions, but
their opinion is skewed
into a direction. This allows us to analyze how phenomena like false balance or biased media can push a population’s average opinion towards a certain direction. 

We provide theoretical upper and lower bounds
on how much one and two external sources can influence the average opinion in
the network (see \cref{thm:twomediaoneround}). We then argue that our bounds are tight. More precisely, we show that for regular graphs, our upper and lower bounds match, i.e., that for regular graphs, our analysis is exact. Moreover, we experimentally observe that on real-world social networks (SN), such as Facebook, and synthetic graph models, such as Barabási-Albert (BA) graphs~\cite{barabasi_PA_graph}, the growth/decrease of the average opinion and the time it takes to reach either radicalized average opinions (close to~$0$ or close to~$1$) is similar to the theoretical bounds provided in~\cref{thm:twomediaoneround} and \cref{lem:numroundsphase1multimedia}.

We also study a setting of \emph{repeated} influence
from two external sources, in which we consider multiple periods of
opinion convergence, after each of which the two external sources update their
bias. In the case of two media sources in which one is stronger than the other,
we show that in regular graphs even a constant number of periods suffices to
radicalize almost all opinions in the network (see \cref{lem:numroundsphase1multimedia}). Interestingly, we experimentally show that in the regular graphs, this discrepancy in strength between the media sources can be as small as one media source being connected to one more node than the other media source to achieve this radicalization of opinions.
For two equally influential external sources, we specify the nodes' final opinions
exactly for regular graphs (see \cref{lemma:twomediaoneroundalphahalf}). In particular, this proposition states that in this setting the total sum of opinions stays unchanged over time, which we additionally verify experimentally.

Furthermore, we give results that differentiate between stubborn and
non-stubborn external sources. Here, we say that an external source is
\emph{stubborn} if it does not participate in the opinion formation process, and
it is \emph{non-stubborn} if it updates its opinion like any other node. We show
that non-stubborn external sources can have only very small impact, whereas the
impact of stubborn external sources (that we study in the rest of the paper) is
significantly higher (see \cref{prop:singlemedia-nonfixed}).
This suggests that it is essential that the media face public scrutiny and peer pressure, to avoid them from (deliberately or unintentionally) biasing a population's average opinion.

\subsection{Related Work} 

Studying opinion formation models and their properties has been an active area
of research for at least two decades in the computer science literature. Some of
the most well-established models are the threshold
model~\cite{kempe2003maximizing,tran2022heterogeneous}, the majority
model~\cite{brill2016pairwise,zehmakan2021majority,grandi2023identifying}, and
the voter model~\cite{petsinis2023maximizing}. In this paper, we focus on the popular FJ model.%

The paper most closely related to ours is by Gionis, Terzi and Tsaparas \cite{gionis2013opinion}, who
considered the problem of identifying the best $k$ individuals in a network such
that if their expressed opinions are fixed to $1$, the sum of expressed opinions
is maximized; this was motivated, e.g., by marketing campaigns.
They found that this problem is NP-hard, but since the objective function is
monotone and submodular, this problem admits a greedy $(1 - \frac{1}{e})$-approximation algorithm. In this paper, we also consider the sum of opinions in
a network and fix the expressed opinions of some nodes. The main difference is
that the results in \cite{gionis2013opinion} are algorithmic, whereas here we are
interested in obtaining analytic bounds on how much the opinions can change.
Therefore, the techniques developed by Gionis, Terzi and Tsaparas do not apply in our
setting.

Abebe et al. \cite{abebe2018opinion} also studied a variation of the FJ model in which each
individual has a resistance or stubbornness parameter measuring the individuals'
propensity for changing their opinion. They consider the problem of how to
change the individual's resistances to maximize (or minimize) the sum of
opinions.

Musco, Musco and Tsourakis \cite{musco2018minimizing} considered the problem of minimizing
the polarization--disagreement index in the FJ model and showed that their objective
function is convex which yields an exact polynomial-time
algorithm. Zhu, Bao and Zhang \cite{zhu2021minimizing} studied a similar problem with the goal of
adding a small number of edges to the network, and showed that the
objective function of their problem is not submodular, although it is monotone.
A similar problem based on changing the opinions of a small number of nodes was
considered by Makos, Terzi and Tsaparas \cite{matakos2017measuring}.

Chen and Racz \cite{chen2020network} and Gaitonde, Kleinberg and Tardos \cite{gaitonde2020adversarial} considered the impact
of adversaries on the FJ model, who aim to maximize the polarization and
disagreement in the network. They derived analytic bounds on the maximum
possible impact of adversaries and also considered the underlying algorithmic
problems. This was extended to adversaries with limited
information by Tu, Neumann and Gionis ~\cite{tu2023adversaries}.

Recently, several works have empirically analyzed the influence of an external media source in different opinion dynamics models. Crokidakis \cite{crokidakis2012effects} and Nazeri \cite{nazeri2018effect} considered the influence of such external mass media by means of Monte Carlo simulations in the two-dimensional Sznajd model \cite{sznajd2000opinion} and Muslim et al. \cite{muslim2024mass} in the voter model. Pineda and Buendía \cite{PINEDA201573} experimentally analyze the effect of an external media source in the Hegselmann and Krause model \cite{rainer2002opinion}. Lastly, Auletta, Coppola and Ferraioli \cite{auletta2021impact} and Candogan \cite{candogan2022social} considered the setting in which the social media platform itself gives news recommendations to the users, targeting at maximizing user activity on the platform, in the DeGroot model~\cite{DeGroot}. %

Apart from external media sources, different forms of (external) bias in opinion dynamics have been considered. The inclusion of stubborn agents (agents with a bias towards a specific opinion) and zealots (agents that never deflect from their initial opinion) has by considered extensively (cf. ~\cite{auletta2017information,mobilia2007role,galam2007role}). Moreover, the scenario in which there is a bias towards a certain (superior) alternative is studied in the majority by Anagnostopoulos \cite{anagnostopoulos2022biased} and in the voter model by Berenbrink et al. \cite{berenbrink2016bounds}. Lastly, Wilder and Vorobeychik \cite{wilder2018controlling} and Corò et al. \cite{coro2022exploiting} consider a variant of influence maximization problem in the Independent Cascade and Linear Threshold model respectively, where they have the budget to convince $k$ nodes initially of some preferred opinion. These $k$ nodes in turn start spreading messages which biases the nodes who receive the message towards this preferred opinion. 

\section{Preliminaries}

\spara{Graph Notation.} Throughout this paper, we let $G=(V, E, w)$ be an
undirected weighted graph which represents a social network. We set $n = \abs{V}$
and $w \colon E \rightarrow \R_{>0}$. %
For a node $i\in V$, $N\left(i\right):=\{v\in V: \{v,i\} \in E\}$ is the
\emph{neighborhood} of $i$, and $d_i:= \sum_{j\in N(i)} w_{ij}$ is the
\emph{degree} of $i$. We let $d_{\max}$ and $d_{\min}$ denote the
maximum degree and minimum degree in $G$, respectively. The weighted adjacency
matrix $\Adj \in \mathbb{R}^{n\times n}$ is defined as $W_{ij}=w_{ij}$ for all
$i,j\in V$. We let $\Diag \in \mathbb{R}^{n\times n}$ be the diagonal degree
matrix given by $D_{ii} = \sum_{j\in N(i)}w_{ij}$ and $0$ off the diagonal.
We let $\laplacian := \Diag - \Adj$ denote the graph Laplacian. We let
$\ID \in \mathbb{R}^{n\times n}$ be the identity matrix, and $\one \in
\mathbb{R}^{n}$ be the vector with $1$ in each entry. All logarithms are natural logarithms unless mentioned otherwise.

\spara{Friedkin--Johnsen opinion dynamics.}
We study opinion dynamics based on the Friedkin--Johnsen (FJ)
model~\cite{friedkin1990social}. The dynamics are specified by a graph
$G=(V,E,w)$, where each node~$i\in V$ has a fixed (private) internal
opinion~$\ebegop_i\in[0,1]$ and a (public) expressed
opinion~$\efinop_i^{(t)}\in[0,1]$, which depends on the time $t\in\mathbb{N}_0$.
Intuitively, one can think of $[0,1]$ as an interval of opinions where $0$ and
$1$ are the two extreme viewpoints, for instance, $0$ corresponds to
the viewpoint that climate change is the most pressing issue of the world and $1$
corresponding to denying climate change. It will be convenient for us to
consider the vectors $\begop \in [0,1]^n$ and $\finop^{(t)}\in[0,1]^n$ of
innate and expressed opinions. Starting with $\finop^{(0)}=\begop$, at each
time step~$t$ all nodes $i \in V$ update their expressed opinions by taking the
weighted average of their neighbors' expressed opinions, as well as their own
innate opinion:
\begin{equation}
\label{eq:FJdynamics}
        \efinop{i}^{(t+1)} = \frac{\ebegop{i} + \sum_{j\in N(i)}w_{ij}\efinop{j}^{(t)}}{1 + \sum_{j\in N(i)}w_{ij}}.
    \end{equation} 
It is worth noticing that in this update rule, we implicitly assume that each
node's innate opinion has unit weight~$1$. The above update rule can be equivalently expressed as:
\begin{equation}
	\label{equation:FJ-one}
	\finop^{(t+1)} = (\ID + \Diag)^{-1}\left(\begop + \Adj \finop^{(t)}\right). 
\end{equation}
It is well-known that for $t\to\infty$, the vector of expressed
\emph{equilibrium}
opinions is given by
\begin{equation}
	\label{equation:FJ-final-one}
	\finop^{*} := \lim_{t\to\infty} \finop^{(t)} = (\ID + \laplacian)^{-1} \begop.
\end{equation}
Interestingly, the sum of innate and of equilibrium opinions is the same (\cref{cor:suminitialsumfinal}). This is well-known in the literature.
\begin{lemma}\label{cor:suminitialsumfinal}
   It holds that $\one^{\intercal} \finop^{*} = \one^{\intercal} \begop$.
\end{lemma}

\spara{Convergence properties.} 
We will study several convergence properties of our model, where we apply some additional notation and several lemmas from \cite{berman1994nonnegative} for our proofs. 
Specifically, we will use the characterizations of \emph{nonhomogeneous first-order matrix difference equations}~(Definition~\ref{def:support:matrix-difference}) and properties of \emph{$M$-matrices}~(Definition~\ref{def:support:m-matrix}).

\begin{definition}
\label{def:support:matrix-difference}
    A \emph{nonhomogeneous first-order matrix difference equation} is of the form $\vx^{t+1} = \cH \vx^{t} + \cc$, where $\cH$ is an $n \times n$ matrix, $\cc$ is an $n \times 1$ constant vector, and $\{\vx^{t}\}_{t=1}^{\infty}$ is an infinite sequence of $n \times 1$ vectors. 
\end{definition}
The following lemma characterizes the solution and convergence properties of nonhomogenous first-order matrix difference equations~\cite[Chapter 7, Lemma 3.6]{berman1994nonnegative}.
\begin{lemma}
\label{lem:support:convergence}
Let $\cA = \cM -\cN \in \mathbb{R}^{n\times n}$ be such that both $\cA$ and $\cM$ are nonsingular matrices. 
Let $\cH = \cM^{-1} \cN$ and $\cc = \cM^{-1} \cb$.
The vector $\{\vx^{t}\}_{t=1}^{\infty}$ of the nonhomogenous first-order matrix difference equation $\vx^{t+1} = \cH \vx^{t} + \cc$ 
converges if and only if $\rho(\cH) <1$. 
Here $\rho(\cH)$ denotes the spectral radius of $\cH$.
Moreover, $\lim_{t\rightarrow \infty} \vx^{t} = \cA^{-1} \cb$. 
\end{lemma}

\begin{definition}[$M$-matrix]
\label{def:support:m-matrix}
    Let $\cA$ be an $n\times n$ matrix where $a_{i,j} \leq 0$ for all $i \neq j$. 
    Then $\cA$ is an $M$-matrix if it is positive semidefinite, i.e., for any vector $\vx$, it holds that $\vx^{\intercal} \cA 
    \vx \geq 0$.
\end{definition}

Note that the graph Laplacian matrix $\laplacian$ is an $M$-matrix,  as it is positive semidefinite~\cite[Proposition 1]{von2007tutorial} and its off-diagonal elements are non-positive. 
Next,  we provide two properties of a nonsingular $M$-matrices~ \cite[Chapter 6, Theorem 2.3]{berman1994nonnegative}. 
\begin{lemma}
    \label{lem:support:m-matrix:equivalent-class}
    The following two properties hold: 
    \textbf{(a)} If $\cA$ is an $M$-matrix, then $\cA + \cD$ is a nonsingular $M$-matrix for each positive diagonal matrix $\cD$.
    \textbf{(b)} If $\cA$ is a nonsingular $M$-matrix, then it is inverse-positive; that is $\cA^{-1}$ exists and $\cA^{-1} \geq 0$, where the inequality holds elementwise.  
\end{lemma}

\section{Stubborn media sources with one period}
\label{sec:StubbornMediaSources}

In this section, we examine the impact of stubborn media sources on opinion formation among individuals in social networks. These media sources are
considered stubborn as they keep their (expressed) opinions fixed throughout an entire
period of opinion dynamics, i.e., they do not adhere to the FJ-dynamics in \cref{eq:FJdynamics}. 
We first give an overview of the model and introduce some further notation in \cref{sec:StubbornMediaSourcesSetting}. %

\subsection{Our model}
\label{sec:StubbornMediaSourcesSetting}
We consider the scenario in which two competing external media sources $M$ and $M'$ are added to the social network. In real-world scenarios, $M$ and $M'$ could be the media sources with opposing political standings, like CNN and Fox News.

Formally, we let $G=(V, E, w)$ be an undirected
weighted graph, and $M$ and $M'$ be two external media sources. %
Each node $i \in V$ is either connected to $M$ or to $M'$, i.e., we either add
an edge $(i, M)$ or $(i, M')$ for all $i\in V$. We set the weight of this edge
to $\beta(1 + d_i)$, where $\beta\geq0$ is a model parameter.
This means that an external media source contributes a fraction of $\beta$ to
the influence on the users, and each node is exclusively connected to one media
source. Additionally, we let $\alpha \in [0, 1]$ denote the fraction of nodes
connected to $M$. Hence, there are $\alpha n$ nodes
connected to $M$, while $(1 - \alpha)n$ nodes are connected to $M'$. 
When $\alpha = 0$ or $\alpha = 1$, the nodes are influenced by a single media
source, which we refer to as the \emph{single media setting}.

In this section, we assume that the expressed
opinions of $M$ and $M'$ are fixed for the whole period, 
i.e., $M$ and $M'$ do not participate in the
opinion dynamics from \cref{eq:FJdynamics}.
We use $\efinop{M}\in[0,1]$ to denote the fixed expressed opinion of $M$ and
$\efinop{M'}\in[0,1]$ to denote the expressed opinion of $M'$.
For convenience, we define a vector $\zeta$ such that $\zeta_i = \efinop{M}$ if node $i$ is adjacent to $M$, and $\zeta_i = \efinop{M'}$ if node $i$ is adjacent to $M'$.

In our analysis, we use $\md{\finop}^{(t)}$ to denote the vector of expressed
opinions at time step $t$ of nodes in the graph $G$ that contains $M$ and $M'$.
We define $\md{\finop}^* = \lim_{t\to\infty} \md{\finop}^{(t)}$. We start by giving a closed-form formulation for the equilibrium expressed
opinions $\md{\finop}^*$.
\begin{theorem}
\label{lem:equilibrium-two-media-sources}
Let $G=(V,E,w)$ be a weighted graph and consider the setting described in \cref{sec:StubbornMediaSourcesSetting}, the equilibrium expressed opinions $\md{\finop}^*$ can be formulated as:
    \begin{equation}
        \md{\finop}^* = ((1 + \beta)\ID + \beta \Diag + \laplacian)^{-1}(\begop + \beta\left(\ID +\Diag\right) \zeta).
    \end{equation}\label{eq:multiplemediazstar}
\end{theorem}
Before providing the proof of this theorem, it is worth noticing that \cref{lem:equilibrium-two-media-sources} is a
generalization of \cref{equation:FJ-final-one}: in the absence of external
sources, indicated by $\beta = 0$, \cref{lem:equilibrium-two-media-sources}
matches \cref{equation:FJ-final-one}. We note that the result of the theorem
holds irrespective of the concrete values of $\efinop{M}$ and $\efinop{M'}$.

Intuitively, \cref{lem:equilibrium-two-media-sources} shows that the changes to the graph topology in the
graph by including $M$ and $M'$ can be translated to adapting and slightly
reweighting the innate opinions of the nodes in the initial graph, which we also illustrate in
\cref{fig:side_by_side_plots}.

\begin{figure*}[t]
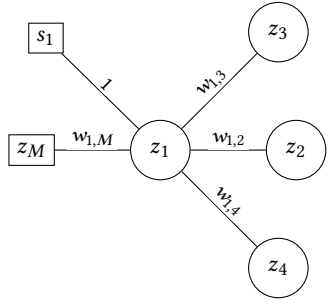
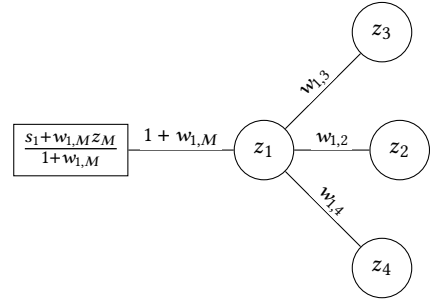

  \begin{subfigure}[b]{0.43\textwidth}
    \centering
    \inputtikz{tikz/fig1}
    \caption{An example of how $z_1$ gets influence from different nodes,including $M$, in our model.}
    \label{fig:first_plot}
  \end{subfigure}
  \hspace{1em}
  \begin{subfigure}[b]{0.43\textwidth}
    \centering
    \inputtikz{tikz/fig2}
    \caption{An equivalent influence on $z_1$, by merging the influence of node~$M$
    with the node $1$'s innate opinion~$s_1$.
  }
    \label{fig:second_plot}
  \end{subfigure}
  \caption{Two equivalent ways to present the influence of a stubborn media
	  source $M$ and its neighbors on node $1$, at each time step. 
        The nodes represent the innate
	  or expressed opinions; we use circles to present nodes' expressed
	  opinions and use boxes to annotate fixed innate opinions. We assume that
	  $N(1) = \{2,3,4\}$ and $w_{1,M} = \beta(1 + \sum_{i=1}^3 w_{1,i})$.}
  \label{fig:side_by_side_plots}
\end{figure*}

\begin{proof}[Proof of Theorem~\ref{lem:equilibrium-two-media-sources}]
Let $\md{\eDiag}_{ii} = \eDiag{i}{i} + \beta(1 + \eDiag{i}{i})$.
By expanding \cref{eq:FJdynamics}, node $i$'s expressed opinion at time step $t+1$ can be formulated as follows:
\begin{equation*}
    \begin{split}
    \md{\efinop}_i^{(t+1)} &= \frac{\ebegop{i} + \sum_{j\in N(i)}w_{i,j}\md{\efinop}_{j}^{(t)} + \beta(1 + \sum_{j\in N(i)}w_{i,j}) \zeta_i}{1 + \sum_{j\in N(i)}w_{i,j} + \beta(1 + \sum_{j\in N(i)}w_{i,j})} \\
    &= \frac{\sum_{j\in N(i)}w_{i,j}\md{\efinop}_{j}^{(t)}}{1 + \md{\eDiag}_{ii}} + \frac{\ebegop{i}}{1 + \md{\eDiag}_{ii}} + \frac{\beta(1 + \eDiag{i}{i}) \zeta_i}{1 + \md{\eDiag}_{ii}}.\\
    \end{split}
\end{equation*}  
Writing this in matrix notation, we get: 
\begin{align*}
    \md{\finop}^{(t+1)} = (\ID + \md{\Diag})^{-1} \Adj \md{\finop}^{(t)} + (\ID + \md{\Diag})^{-1} (\begop + (\md{\Diag} - \Diag)  \zeta).
\end{align*}
Note that this is a nonhomogeneous first-order matrix difference equation (see \cref{def:support:matrix-difference}), and hence we can apply \cref{lem:support:convergence} to check whether $\md{\finop}^{(t)}$ converges and obtain the value of $\lim_{t\rightarrow \infty} \md{\finop}^{(t)}$.

We apply \cref{lem:support:convergence} with $\cH = (\ID + \md{\Diag})^{-1} \Adj$ and $\cc = (\ID + \md{\Diag})^{-1} (\begop + (\md{\Diag} - \Diag) \zeta)$. 
Additionally, we set $\cM$ to the common term of $\cc$ and $\cH$, i.e., we let $\cM = \ID + \md{\Diag}$; hence, $\cN = \Adj$, $\cb = \begop + (\md{\Diag} - \Diag) \zeta$ and $\cA = \ID + \md{\Diag} - \Adj$. 

Next, we check whether $\cA$ and $\cM$ are singular. 
We notice that $\cA$ is a nonsingular $M$-matrix, as \laplacian is an $M$-matrix and $\cA = \ID + \md{\Diag} - \Adj = (1+\beta)\ID + \beta\Diag + \laplacian$ is the sum of an $M$-matrix and a positive diagonal matrix. 
Hence by \cref{lem:support:m-matrix:equivalent-class}, $\cA$ is a nonsingular $M$-matrix. 
Furthermore, $\cM$ is nonsingular as it is a positive diagonal matrix. 
Hence, we can directly apply \cref{lem:support:convergence}, and get:
\begin{align*}
\md{\finop}^{*} = \cA^{-1} \cb &= (\ID + \md{\Diag} - \Adj)^{-1} [\begop + (\md{\Diag} - \Diag) \zeta]\\
&= ((1+\beta)\ID + \beta\Diag + \laplacian)^{-1} (\begop + \beta (\ID+\Diag) \zeta).
\qedhere
\end{align*}
\end{proof}

\para{Competing media sources.}
For the rest of this paper, we assume that the two media sources exert opposing
influences on the average opinions of network users: $M$~aims to increase the
average opinion of network users (i.e., bring it closer to $1$), while $M'$ aims
to decrease the average opinion (bring it closer to $0$).
To this end, we consider a bias parameter $\gamma\in(0,1)$ and let $\bar{\begop} =
\frac{\one^{\intercal}\begop}{n}$ denote the average innate opinion of the network users. 
We then set the expressed opinion $\efinop{M}$ of media source $M$ to
$\efinop{M}=\min \{(1+\gamma) \bar{\begop}, 1\}$, and the expressed opinion
$\efinop{M'}$ of media source $M'$ to $\efinop{M'}=(1-\gamma) \bar{\begop}$. It
is important to note that there are instances where
$\efinop{M} = (1+\gamma) \bar{\begop}$ might be greater than $1$, thus exceeding
the maximum value; we therefore add the minimum constraint. 
We do not need to set the similar constraint on $\efinop{M'}$, as $(1-\gamma) \bar{\begop}$ 
is always at least $0$.

\subsection{Sum of opinions}
\label{sec:StubbornMediaSourcesSingleSum}
Our goal is to bound how much impact the external sources $M$ and $M'$ can
have on the average opinion in the network. To this end, we derive a bound on
how much the total sum of opinions in the network \emph{with} external media
sources differ from the total sum of opinions in the network \emph{without}
external media sources.

Recall that above we set $\efinop{M}=\min\{(1+\gamma) \bar{\begop},1\}$. We obtain
different results for the two cases and refer to them as \emph{non-truncated} and \emph{truncated} opinions.

Before we present our main results of these two cases, let us introduce
\cref{lem:sum-opinion-helper}. 
Our main technical results in this section are built on the top of this lemma since it allows us to derive a bound on the sum of
opinions that purely relies on the maximum and minimum degree of the graph. 
We provide the proof
of this lemma in \cref{sec:proof-thm-twomediaoneround}. 

\begin{lemma}\label{lem:sum-opinion-helper}
Let $G=(V,E,w)$ be a weighted graph, with degree matrix $\Diag$ and Laplacian matrix $\laplacian$.
Let $\beta\in[0, 1]$, and $d_{\min}$ and $d_{\max}$ be the minimum and maximum degree of $G$, respectively.
Then it holds that,
\begin{enumerate}
\item $\one^\intercal ((1 + \beta)\ID + \beta \Diag + \laplacian)^{-1} \leq \frac{1}{\beta(d_{\min} +1) + 1} \one^{\intercal}.$
\item $\one^\intercal ((1 + \beta)\ID + \beta \Diag + \laplacian)^{-1} \geq \frac{1}{\beta(d_{\max} +1) + 1}\one^{\intercal}.$
\end{enumerate}
The inequality holds element-wise. Furthermore, this holds with equality when the graph is $d$-regular. 
\end{lemma}

\para{Non-truncated opinions.} 
First, we present a result in
\cref{thm:twomediaoneround} if $\efinop{M} = (1+\gamma) \bar{\begop}$, i.e., the opinion of $\efinop{M}$ was not truncated.

\begin{theorem}\label{thm:twomediaoneround}
Let $G=(V,E,w)$ be a weighted graph and consider the setting described in
\cref{sec:StubbornMediaSourcesSetting} with $\efinop{M} = (1+\gamma)
\bar{\begop}$ and $\efinop{M'} = (1-\gamma) \bar{\begop}$. Then the sum of
equilibrium expressed opinions can be bounded as,
\begin{align*}
 \begin{array}{rcl}
    \one^\intercal \md{\finop}^* & \leq & \frac{1 + (d_{\max} + 1) \beta ((2\alpha -1)\gamma + 1) }{\beta(d_{\min}+1)+1} \one^{\intercal} \begop,  \\
    \one^\intercal \md{\finop}^* & \geq & \frac{1 + (d_{\min} + 1) \beta ((2\alpha -1)\gamma + 1) }{\beta(d_{\max}+1)+1} \one^{\intercal} \begop.
    \end{array}
\end{align*}
\end{theorem}

To better illustrate the bounds,
\cref{cor:twomediaoneroundreg} states the bound for $d$-regular graphs, for
which it is \emph{tight}.

\begin{corollary}
    \label{cor:twomediaoneroundreg}
	Let $G=(V,E,w)$ be a weighted $d$-regular graph, i.e., $d_{\max} =
	d_{\min} = d$ and $\efinop{M} = (1+\gamma) \bar{\begop}$ and $\efinop{M'} =
	(1-\gamma) \bar{\begop}$.
	Then the sum of equilibrium expressed opinions is given by:
    \begin{align*}
        \one^\intercal\md{\finop}^* = \left(1 + \gamma \cdot \frac{\beta (d+1)  (2\alpha -1)}{\beta (d+1) + 1}\right) \one^{\intercal} \begop.
    \end{align*}
\end{corollary}

Observe that the sum of opinions increases
(decreases) when $\alpha > \frac{1}{2}$ ($\alpha < \frac{1}{2}$). Furthermore,
when $\beta (d+1)$ is sufficiently larger than $1$ then the theorem states that
$\one^\intercal\md{\finop}^* = (1 + \Omega(\gamma (2\alpha -1))) \one^{\intercal} \begop$,
i.e., under these conditions, the sum of opinions is only controlled by the bias
parameter~$\gamma$ and how much the stronger media source dominates
(controlled by $\alpha$).

This result is significant for two reasons: (1)~\cref{cor:suminitialsumfinal}
asserts that \emph{without} the intervention of external media sources, the total sum
of expressed and innate opinions of the network users always stays the same,
i.e., $\one^\intercal \finop^* = \one^\intercal \begop$. Hence, the theorem characterizes the power of the external sources and shows that, under the
parameter settings from above, the external sources bias the average opinion by a factor of
$1 + \Omega(\gamma (2\alpha -1))$.
(2)~The result of the theorem is enabled by the fact that the media sources $M$
and $M'$ are stubborn, i.e., their expressed opinions are fixed. In \cref{sec:non-stubborn} we show that if there is only a single
media source (i.e., $\alpha=1$), which can only control its innate opinion but
has to update its expressed opinion based on the update rule in
\cref{eq:FJdynamics}, the bias on the sum of opinions is much weaker (it only contributes a factor of $1+\frac{1+\gamma}{n}$ rather than $1+\Omega(\gamma)$).

\begin{proof}[Proof of \cref{thm:twomediaoneround}]
First we give a lower bound on the sum of expressed opinions, 
\begin{align*}
    \one^\intercal\md{\finop}^* &\stackrel{\mathclap{(a)}}{=} \one^\intercal\left((1 + \beta)
    \cdot \ID + \beta \cdot \Diag + \laplacian\right)^{-1}\left(\begop + \beta(\ID + \Diag)\zeta \right)\\
    &\stackrel{\mathclap{(b)}}{\geq}\frac{1}{\beta(d_{\max} + 1) + 1}\one^{\intercal}(\begop + \beta(\ID + \Diag)\zeta )\\
    &\geq \frac{1}{\beta(d_{\max} + 1) + 1}\left(\one^{\intercal}\begop + (d_{\min}+1)\beta \one^{\intercal}\zeta \right)\\
    &\stackrel{\mathclap{(c)}}= \frac{1 + (d_{\min} + 1) \beta ((2\alpha -1)\gamma + 1) }{\beta(d_{\max}+1)+1} \one^{\intercal} \begop,
\end{align*} 
where in Step~(a) we used \cref{lem:equilibrium-two-media-sources}, 
in Step~(b) we used \cref{lem:sum-opinion-helper} (2), 
and Step~(c) follows from how we set $\zeta$ (recall that 
$\zeta_i = \efinop{M} = (1+ \gamma) \cdot \frac{\one^{\intercal} \begop}{n}$ if
$i$ is adjacent to $M$, and 
$\zeta_i = \efinop{M'}= (1 - \gamma) \cdot \frac{\one^{\intercal} \begop}{n}$ if
$i$ is adjacent to $M'$). To obtain the upper bound stated in \cref{thm:twomediaoneround}, we perform the same calculation, except that now we
use \cref{lem:sum-opinion-helper} (1) in Step~(b).
\end{proof}

\para{Truncated opinions.}
Now, we consider
the case when $(1+\gamma) \bar{\begop} > 1$, and hence we truncate $\efinop{M}$ by setting $\efinop{M} = 1$. Nonetheless, we still set
$\efinop{M'} = (1 - \gamma)\frac{\one^{\intercal} \begop}{n}$. Note that this setting is interesting, since now $M$ induces a smaller bias than before, and we need to understand how much this boosts the impact of $M'$.
For $d$-regular graphs, we obtain \cref{lem:multimediaincreasePHASE2}, which
gives us a closed form solution for $\one^\intercal\md{\finop}^{*}$.

\begin{proposition}\label{lem:multimediaincreasePHASE2}
Let $G=(V,E,w)$ be a weighted $d$-regular graph. Suppose 
$\efinop{M} = 1$ and $\efinop{M'} = (1-\gamma) \bar{\begop}$. Then the sum of equilibrium expressed opinions is
\begin{equation}
    \label{eq:lem:multimediaincreasePHASE2}
    \one^\intercal\md{\finop}^{*} = 
	\frac{(1 + \beta (1 + d) (1-\alpha) (1-\gamma)) \one^{\intercal} \begop + \alpha \beta (1+d) n}{1 + \beta(1+d)}.
\end{equation}
\end{proposition}
Observe that the proposition implies that  $\one^\intercal\md{\finop}^{*}$ increases linearly as a function of $\alpha$
(since $\one^{\intercal} \begop \leq n$), and it decreases linearly as a function of
the media bias $\gamma$. This dependency on $\gamma$ stems from the fact that we truncate $\efinop{M}$, while \efinop{M'} decreases as $\gamma$ increases.

Furthermore, in \cref{cor:multimediaincreasePHASEtwoalphadbound} we derive a lower bound on $\one^\intercal \md{\finop}^{*}$ that is
\emph{independent} of $d$ and $\beta$ that purely relies on the innate opinions,
$\alpha$ and $\gamma$. Note that this provides a general bound on the power of the bias of $M'$ if $\efinop{M}=1$, since we provide a lower bound on $\one^\intercal\md{\finop}^{*}$.

\begin{corollary}\label{cor:multimediaincreasePHASEtwoalphadbound}
Suppose $G$ is $d$-regular.
By rearranging \cref{eq:lem:multimediaincreasePHASE2}, we obtain the following
lower bound on $\one^\intercal\md{\finop}^{*}$:
    $\one^\intercal\md{\finop}^{*} > \one^\intercal \begop (1 - \gamma + \alpha \gamma).$
\end{corollary}

\begin{proof}[Proof of \cref{lem:multimediaincreasePHASE2}]

We use \cref{lem:equilibrium-two-media-sources} and simplify the expression:
\begin{align*}
    \one^{\intercal} \md{\finop}^* &= \one^{\intercal} ((1+\beta) \ID + \beta \Diag + \laplacian)^{-1} (\begop + \beta (\ID + \Diag) \zeta) \\
    &\stackrel{(a)}{=} \frac{1}{ 1 + \beta (1+d) } \one^{\intercal} (\begop + \beta (\ID + \Diag) \zeta) \\
    &\stackrel{(b)}{=} \frac{\one^{\intercal} \begop + \beta (1 + d) (\alpha n + (1-\alpha) (1 - \gamma) \one^{\intercal} \begop)}{1 + \beta (1+d) } \\
    &= \frac{1 + (1-\alpha) (1 - \gamma) \beta (1+d) }{1 + \beta(1+d) } \one^{\intercal} \begop + \frac{\alpha n \beta (1 + d)}{1 + \beta (1+d) },
\end{align*}
where
Step~(a) holds by \cref{lem:sum-opinion-helper}, which we prove below.
Step~(b) holds by plugging $\efinop{M} = 1$ and $\efinop{M'} = (1 - \gamma)\frac{\one^{\intercal} \begop}{n}$ into $\zeta$, since
    $\one^{\intercal} \zeta = \alpha n \cdot 1 + (n - \alpha n) \frac{\one^{\intercal} \begop}{n} (1 - \gamma) 
    = \alpha n \cdot 1 + (1 - \alpha ) (1 - \gamma) \one^{\intercal} \begop. 
    $

\end{proof}

\begin{proof}[Proof of \cref{cor:multimediaincreasePHASEtwoalphadbound}]
We start by applying \cref{lem:multimediaincreasePHASE2}:

    \begin{align*}
    &\frac{1 + (1-\alpha) (1 - \gamma) \beta (1+d) }{1 + \beta(1+d) } \one^{\intercal} \begop + \frac{\alpha n \beta (1 + d)}{1 + \beta (1+d) }\\
    &\stackrel{(a)}{>} \frac{1 + (1-\alpha) (1 - \gamma) \beta (1+d) + \alpha \beta (1 + d)}{1 + \beta(1+d) } \one^{\intercal} \begop \\
    &\stackrel{(b)}{\geq} \frac{(1 - \gamma + \alpha \gamma) + (1-\gamma + \alpha \gamma) \beta (1+d) }{1 + \beta(1+d) } \one^{\intercal} \begop\\
    &=(1 - \gamma + \alpha \gamma)\one^{\intercal} \begop.
\end{align*}
 Note that Step~(a) is obtained from $n > \one^{\intercal} \begop$, and Step~(b) is obtained from $\alpha \leq 1$, hence $1 - \gamma + \alpha \gamma = 1 + \gamma (\alpha -1) \leq 1$. By substituting back into the formula, we obtain 
\begin{align*}\one^{\intercal} \md{\finop}^* > \one^{\intercal} \begop (1 - \gamma + \alpha \gamma),
 \end{align*}
 which is what we claimed in the corollary.
\end{proof}

\subsection{Proof of \cref{lem:sum-opinion-helper}}
\label{sec:proof-thm-twomediaoneround}

    First, we prove that it holds element-wise that 
    $\one^\intercal \geq \frac{1}{\beta(d_{\max}+1)+1} \one^{\intercal} (\ID + \laplacian + (\Diag + \ID) \beta)$, 
    and $\one^\intercal \leq \frac{1}{\beta(d_{\min} +1) +1} \one^{\intercal} (\ID + \laplacian + (\Diag + \ID) \beta) $.
    After that, we show that $(\ID + \laplacian + (\Diag + \ID) \beta)$ is a non-singular $M$-matrix;
    this implies that $(\ID + \laplacian + (\Diag + \ID) \beta)^{-1}$
    only contains nonnegative elements (see \cref{lem:support:m-matrix:equivalent-class}). 
    In the end, we show that by combining these two properties, the lemma follows.

    We note that $\one^\intercal \geq  \frac{1}{\beta(d_{\max}+1)+1} \one^{\intercal} (\ID + \laplacian + (\Diag + \ID) \beta)$, as,
    \begin{align*}
        &\frac{1}{\beta(d_{\max}+1)+1} \one^{\intercal} (\ID + \laplacian + (\Diag + \ID) \beta) \\
        &\stackrel{\mathclap{(a)}}{=}\frac{1}{\beta(d_{\max}+1)+1} \one^{\intercal} \left(\ID + (\Diag + \ID) \beta\right) \\
        &\stackrel{\mathclap{(b)}}{\leq} \frac{1}{\beta(d_{\max}+1) + 1} \one^{\intercal}\ID\left(1 + \beta(d_{\max}+1)\right) = \one^{\intercal},
    \end{align*}
    where Step~(a) holds since $\one^{\intercal}\laplacian = \mathbf{0}$ and Step~(b) holds since $\Diag + \ID \leq (d_{\max} + 1) \ID$. 

    Similarly, to show that $\one^\intercal \leq \frac{1}{\beta(d_{\min} +1) +1} \one^{\intercal} (\ID + \laplacian + (\Diag + \ID) \beta) $, we proceed as follows:
    \begin{align*}
        &\frac{1}{\beta(d_{\min} +1)+1 } \one^{\intercal} (\ID + \laplacian + (\Diag + \ID) \beta) \\
        &\stackrel{(a)}= \frac{1}{\beta(d_{\min} +1) +1} \one^{\intercal}(\ID + (\Diag + \ID) \beta) \\
        &\stackrel{(b)}\geq \frac{1}{\beta(d_{\min} +1) +1}\one^{\intercal} \ID (1 + \beta (d_{\min} + 1)) = \one^{\intercal},
    \end{align*}
    where Step~(a) follows from the fact that $\one^{\intercal} \laplacian = \mathbf{0}$, and Step~(b) holds because $\Diag + \ID \geq (d_{\min} + 1)$.
    
    Next, we notice that $(\ID + \laplacian + (\Diag + \ID) \beta)$ is a non-singular $M$-matrix, 
    as it is the sum of the $M$-matrix $\laplacian$ and the positive diagonal matrix $(\ID + (\Diag + \ID) \beta)$,
    according to \cref{lem:support:m-matrix:equivalent-class}, it is a non-singular $M$-matrix.

    Before we finish the proof, we briefly prove an observation. 
    Consider vectors $\va \geq \vb$, where the inequality holds element-wise, and a vector $\vc$ with non-negative entries.
    Now observe that $\va^{\intercal} \vc = \sum_{i} a_i c_i \geq \sum_{i} b_i c_i = \vb^{\intercal} \vc$.

    Finally, we prove the lemma by combining the previous results. 
    To show this, we let $x_j$ denote the $j$'th entry of
    $\one^\intercal (\ID + \laplacian + (\Diag + \ID) \beta)^{-1}$ 
    and we let $y_j$ denote the $j$'th entry of $ \frac{1}{\beta(d_{\max}+1)+1} \one^{\intercal}$. Then we set $\va^{\intercal} = \one^\intercal$, $\vb^{\intercal} = \frac{1}{\beta(d_{\max}+1)+1} \one^{\intercal} (\ID + \laplacian + (\Diag + \ID) \beta)$, and
    we let $\vc$ be the $j$'th column vector of $(\ID + \laplacian + (\Diag + \ID) \beta)^{-1}$. Notice that by \cref{lem:support:m-matrix:equivalent-class} (see above), $\vc$ is a vector with nonnegative entries. 
    Now we obtain that,
    \begin{align*}
     x_j = \one^\intercal \vc &\geq \frac{1}{\beta(d_{\max}+1)+1} \one^{\intercal} (\ID + \laplacian + (\Diag + \ID) \beta) \vc \\
     &= \frac{1}{\beta(d_{\max}+1)+1} \one^{\intercal} \v+d = y_j,
     \end{align*}
     where $\v+d$ is the indicator vector with a $1$ in the $j$'th entry and $0$ in all other entries.
     Since this holds for all $j$, this implies our first bound.
     
     Similarly, letting $\va^{\intercal} = \frac{1}{\beta(d_{\min}+1)+1} \one^{\intercal} (\ID + \laplacian + (\Diag + \ID) \beta)$, $\vb^{\intercal} = \one^\intercal$, and $\vc$ any column of $(\ID + \laplacian + (\Diag + \ID) \beta)^{-1}$, we can prove 
     $\one^\intercal (\ID + \laplacian + (\Diag + \ID) \beta)^{-1} \leq \frac{1}{\beta(d_{\min}+1)+1} \one^{\intercal} $.

\section{Stubborn media sources with multiple periods}
\label{sec:MultiplePeriod}

In this section, we study the impact of two stubborn media sources on the
expressed opinions over a longer time horizon. In
\cref{sec:StubbornMediaSources} we considered the expressed opinions
$\md{\finop}^*$ after $M$ and $M'$ were added to the graph and the expressed
opinions converged. Now we consider multiple \emph{periods} of such convergence
steps (after each of which $\efinop{M}$ and $\efinop{M'}$ are updated) to
understand how quickly opinions can get radicalized towards an extreme opinion after
multiple periods. This setting is applicable when each period of convergence
corresponds to a new societal topic and is similar to a setting studied by
\cite{gaitonde2020adversarial}.

Formally, our model is as follows. We consider a sequence of \emph{periods}
$t=0,1,2,\dots$. At the beginning of each period~$t$, individuals set their
innate opinions to their expressed opinions from the previous period. 
Formally, let $\mdtwo{\efinop}_i^{(t)}$ denote the equilibrium expressed
opinion of node~$i$ at the end of period~$t$. Then in period~$t+1$, all
nodes~$i$ set their innate opinions to $\ebegop_i^{(t+1)} = \mdtwo{\efinop}_i^{(t)}$.
After updating the nodes' innate opinions, the opinions of the external sources
are also updated, now defined as $\efinop{M}^{(t+1)}:= \min \{(1 + \gamma)
\bar{\begop}^{(t+1)}, 1\},$ and $\efinop{M'}^{(t+1)}:= (1 - \gamma)
\bar{\begop}^{(t+1)},$ where $\bar{\begop}^{(t+1)} =
\frac{\one^{\intercal}\begop^{(t+1)}}{n}$ is the average innate opinion at the
start of period~$t+1$. Then we run the model from
\cref{sec:StubbornMediaSources} to obtain $\mdtwo{\efinop}_i^{(t+1)}$. In our analysis, we restrict ourselves to regular graphs and assume $\alpha \geq
1/2$. We will provide results for $\alpha > \frac{1}{2}$ and for
$\alpha = \frac{1}{2}$.

\spara{Unequally strong media sources.}
First, we assume that $\alpha > \frac{1}{2}$, i.e., that source~$M$ is connected
to strictly more nodes than $M'$. We first compute the minimum number of periods it
takes to obtain $\efinop{M}^{(t)} = 1$, which we consider as a criterion for
radicalization, since in this case the average opinion in the network is at least
$1/(1+\gamma)$. We denote this number of periods by $\ell^*$. 

\begin{proposition}
\label{lem:numroundsphase1multimedia}
Let $G$ be a $d$-regular graph and let $\gamma$, $\mdtwo{\finop}^{(0)}$ $\alpha$ and $\ell^*$ be as defined above. We assume that $\alpha > 1/2$. Then,
\begin{equation*}
    \ell^* = \frac{\log\left(\frac{n}{\one^{\intercal}\begop(1+\gamma)}\right)}{\log\left(1+\gamma\cdot\frac{(d+1)\beta(2\alpha-1)}{(d+1)\beta+1}\right)}.
\end{equation*}   
\end{proposition}
Observe that if for the initial innate opinions at period $0$ it
holds that $\one^{\intercal}\begop=\Omega(n)$ and $\gamma\cdot\frac{(d+1)\beta(2\alpha-1)}{(d+1)\beta+1}$ is a constant (which are reasonable assumptions for most scenarios), then $\ell^* = O(1)$. That is, it takes a constant number of periods to radicalize the sum of opinions in the network.

\begin{proof}[Proof of \cref{lem:numroundsphase1multimedia}]
By \cref{cor:twomediaoneroundreg} we note that,
    \begin{align*}
    \one^\intercal\md{\finop}^* &= \left(1 + \frac{(d+1) \beta \gamma (2\alpha -1)}{(d+1) \beta + 1}\right) \one^{\intercal} \begop.
    \end{align*}
    Now, observe that
    \begin{align*}
        \one^\intercal\mdtwo{\finop}^{(k)} &= \left(1 + \frac{(d+1) \beta \gamma (2\alpha -1)}{(d+1) \beta + 1}\right)^k \one^{\intercal} \begop.
    \end{align*}
    As we assume that $\alpha>1/2$, we are interested in finding the period for which,
    \begin{equation*}
        (1 + \gamma)\bar{\mdtwo{\finop}}^{(k)} = (1 + \gamma)\cdot \frac{\one^{\intercal}\mdtwo{\finop}^{(k)}}{n}=1.
    \end{equation*}
    Now elementary calculations show that,
    \begin{align*}
        &(1 + \gamma)\cdot \frac{\one^{\intercal}\mdtwo{\finop}^{(k)}}{n}=1
        \iff k = \frac{\log\left(\frac{n}{\one^{\intercal}\begop(1+\gamma)}\right)}{\log\left(1 + \frac{(d+1)\gamma(2\alpha -1)}{(d+1)\beta+1}\right)}.
		\qedhere
    \end{align*}
\end{proof}

Next, let us consider multiple periods where the expressed opinion of $M$ is truncated, i.e., $\efinop{M} =1$. Here, even though $\alpha > 0.5$, $\one^{\intercal}\mdtwo{\finop}$ might decrease as $\efinop{M'}$ is still $(1 - \gamma)\bar{\begop}^{(t+1)}$ but $\efinop{M'} = 1$ (instead of $(1+\gamma)\bar{\begop}^{(t+1)}$). However, by setting $\alpha > 0.5$, $\md{\finop}^* = \mdtwo{\finop}^{(t+1)}$ and $\begop = \mdtwo{\finop}^{(t)}$ in \cref{cor:multimediaincreasePHASEtwoalphadbound} we observe that the normalized average opinion cannot drop below $\frac{1}{(1 + \gamma)^2}$. 

\spara{Equally strong media sources.} Next, we consider the case when
$\alpha = \frac{1}{2}$, i.e., when both media sources are equally strong. We
give a closed-form solution of the final expressed opinions after an infinite
number of periods. Here, we denote the final expressed opinions by
$\mdtwo{\finop}^{(\infty)} = \lim_{t\rightarrow \infty} \mdtwo{\finop}^{(t)}$.
Furthermore, we let $\zeta^{(0)}$ denote the vector which contains the initial opinions of the sources that the nodes are connected to; more formally, we set $\zeta_i^{(0)} = \ebegop_M^{(0)}$ if node $i$ is connected to source $M$ and $\zeta_i^{(0)} = \ebegop_{M'}^{(0)}$ if node~$i$ is connected to source $M'$.
Then we obtain the following result.

\begin{proposition}\label{lemma:twomediaoneroundalphahalf}
    Let $G=(V, E, w)$ be a weighted $d$-regular graph and let $\alpha =
	\frac{1}{2}$. Then 
    $\one^{\intercal} \mdtwo{\finop}^{(t)} = \one^{\intercal} \mdtwo{\begop}^{(0)}$ for all $t\geq 0$ and
 $\mdtwo{\finop}^{(\infty)} = (\ID + \frac{1}{\beta(1+d)} \laplacian)^{-1} \zeta^{(0)}$. 
\end{proposition}

\cref{lemma:twomediaoneroundalphahalf} shows that when the media sources are equally strong, the sum of opinions does not change (as one might have expected).
Surprisingly, the proposition also implies that $\mdtwo{\finop}^{(\infty)}$ is solely
dependent on the sources' initial opinions $\zeta^{(0)}$, the parameter $\beta$ which determines the impact of the sources on the nodes, and the
graph's adjacency matrix $\Adj$ (which determines the degree $d$ and the
Laplacian matrix $\laplacian$). Interestingly, observe that if $\beta(d+1)$ is large, $\mdtwo{\finop}^{(\infty)}$ is very close to $\zeta^{(0)}$.

\begin{proof}[Proof of \cref{lemma:twomediaoneroundalphahalf}]
    We first plug $\alpha = \frac{1}{2}$ into \cref{cor:twomediaoneroundreg}, and we notice that $\one^{\intercal} \md{\finop}^* = \one^{\intercal} \begop$. 
    As we always set $\begop^{(t+1)}$ to the expressed equilibrium opinions at the end of period $t$, $\mdtwo{\finop}^{(t)}$, it follows that for any period it holds that $\one^{\intercal} \mdtwo{\finop}^{(t+1)} = \one^{\intercal} \mdtwo{\finop}^{(t)}$.
    
    Next, consider the $t$-th period.
    We set $M$'s opinion to $$z_M^{(t)}:= (1 + \gamma)\frac{\one^{\intercal}\md{\finop}^{(t)}}{n} = (1+\gamma) \frac{\one^{\intercal} \begop}{n},$$
    and $M'$'s opinion to $$z_{M'}^{(t)}:= (1 - \gamma)\frac{\one^{\intercal}\md{\finop}^{(t)}}{n} = (1 - \gamma) \frac{\one^{\intercal} \begop}{n}.$$ 
    In other words, $z_{M}^{(t)}$ and $z_{M'}^{(t)}$ stay the same over all the periods.

    Based on our observation, we use the fact that $\zeta$ always stays the same over all the periods and re-formulate the linear equation of \cref{lem:equilibrium-two-media-sources} into the following nonhomogeneous first-order matrix difference equation:
    \begin{align*}
        \mdtwo{\finop}^{(t+1)} &= ((1 + \beta)\ID + \beta \Diag + \laplacian)^{-1}(\mdtwo{\finop}^{(t)} + \beta \left(\ID + \Diag \right) \zeta^{(0)}).
    \end{align*}

    We again apply \cref{lem:support:convergence} to check whether $\mdtwo{\finop}^{(\infty)} = \lim_{t\rightarrow \infty}\mdtwo{\finop}^{(t)}$ exists and obtain the value.

    To apply \cref{lem:support:convergence}, we set 
    $\cH = ((1 + \beta) \ID + \beta \Diag + \laplacian)^{-1}$, 
    $\cc = ((1 + \beta) \ID + \beta \Diag + \laplacian)^{-1} \beta (\ID + \Diag) \zeta^{(0)}$, 
    $\cM = (1 + \beta) \ID + \beta \Diag + \laplacian$,
    $\cN = \ID$, 
    $\cb = \beta (\ID + \Diag) \zeta^{(0)}$, and 
    $\cA =  \beta \ID + \beta \Diag + \laplacian$.

    We observe that both $\cA$ and $\cM$ are non-singular $M$-matrices, hence 
    \begin{align*}
        \mdtwo{\finop}^{(\infty)} &= \cA^{-1} \cb \\
        &= (\beta \ID + \beta \Diag + \laplacian)^{-1} \beta (\ID + \Diag) \zeta^{(0)} \\
        &= (\ID + (\beta (\ID + \Diag))^{-1}\laplacian)^{-1} \zeta^{(0)} \\
        &= \left(\ID + \frac{1}{\beta(1+d)}\laplacian\right)^{-1} \zeta^{(0)}.
    \end{align*}

    Notice that the last equality holds as the graph is a $d$-regular graph.
\end{proof}

\section{Non-stubborn Media Sources}
\label{sec:non-stubborn}
Next, we consider a single \emph{non-stubborn} media source, i.e., now the media
source participates in the FJ-dynamics like any other node. Formally, our model is as follows.
Initially, we set the innate opinion and the expressed opinion of
$M$ to be the same, i.e.,
$\ebegop_{M} = \efinop_{M}^{(0)}= (1 + \gamma)\bar{\begop}$, but now $\efinop{M}^{(t)}$ is updated based on \cref{eq:FJdynamics}. Intuitively, one would expect that here $M$ influences the other
nodes' opinions much less than in the stubborn setting. Indeed, we show that for any
(possibly non-regular) graph, the sum of expressed opinions increases by
\emph{at most} a factor of $1+\frac{1+\gamma}{n}$. This is in stark contrast to
our discussion after \cref{cor:twomediaoneroundreg}, where we argued that
for $d$-regular graphs with $d(\beta+1)\geq 1$, the sum of expressed opinions
increases by \emph{at least} a factor of $1+\Omega(\gamma)$.
\begin{proposition}
\label{prop:singlemedia-nonfixed}
	After the convergence of opinion dynamics, we obtain that
	$\one^\intercal\md{\finop}^*
		\leq \left(1 + \frac{1+\gamma}{n}\right) \one^\intercal \begop.$
\end{proposition}
\begin{proof}
	By \cref{cor:suminitialsumfinal}, the sum of the expressed opinions of all nodes,
	including the external source, is equal to the sum of innate opinions of
	all nodes, namely $\one^\intercal\md{\finop}^* + \md{z}_{M}^* =
	\one^\intercal \begop + s_{M}$.  Re-arranging the equation,
	$\one^\intercal\md{\finop}^* = \one^\intercal \begop + s_{M} -
	\md{z}_{M}^*$.  As $\md{z}_{M}^* \geq 0$, an upper bound on the sum
	of expressed opinions is given by $\one^\intercal\md{\finop}^* \leq
	\one^\intercal \begop + s_{M}$.  Plugging $s_{M} \leq (1+\gamma)\one^\intercal \begop/n$ into the formula, we obtain 
	\begin{align*}
		\one^\intercal\md{\finop}^* &\leq \one^\intercal \begop + (1+\gamma) \one^\intercal \begop/n 
		= \left(1 + \frac{1+\gamma}{n}\right) \one^\intercal \begop.
		\qedhere
	\end{align*}
\end{proof}

\section{Experimental Results}

\begin{figure*}[ht!]
    \centering
    \begin{subfigure}{0.18\textwidth}
     \captionsetup{font=scriptsize}
        \includegraphics[scale=0.22]{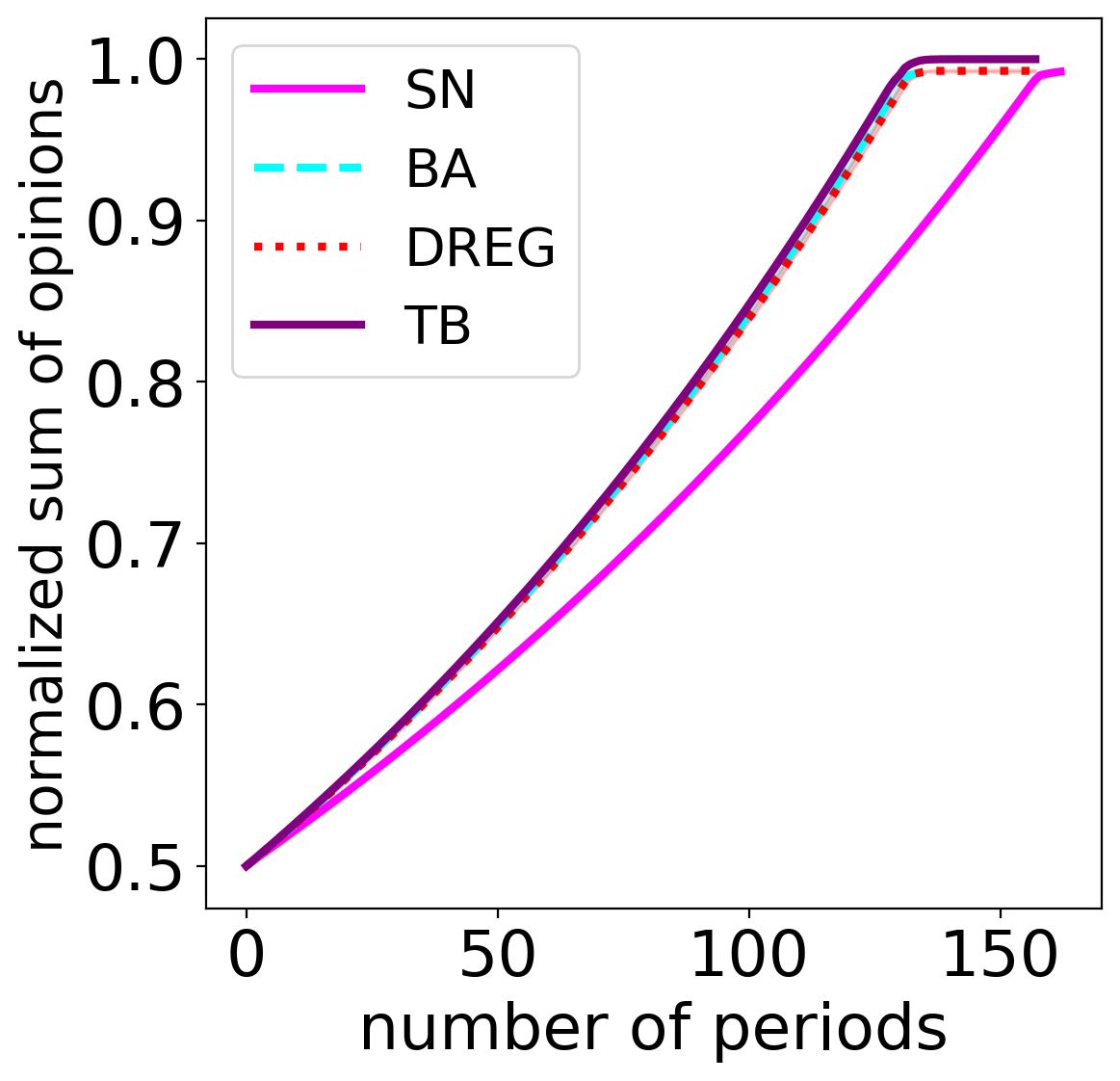}
\caption{$\alpha=1,\gamma=0.01,\beta=0.025$}
        \label{fig:EX1FB}
    \end{subfigure}
    \begin{subfigure}{0.18\textwidth}
    \captionsetup{font=scriptsize}
        \includegraphics[scale=0.22]{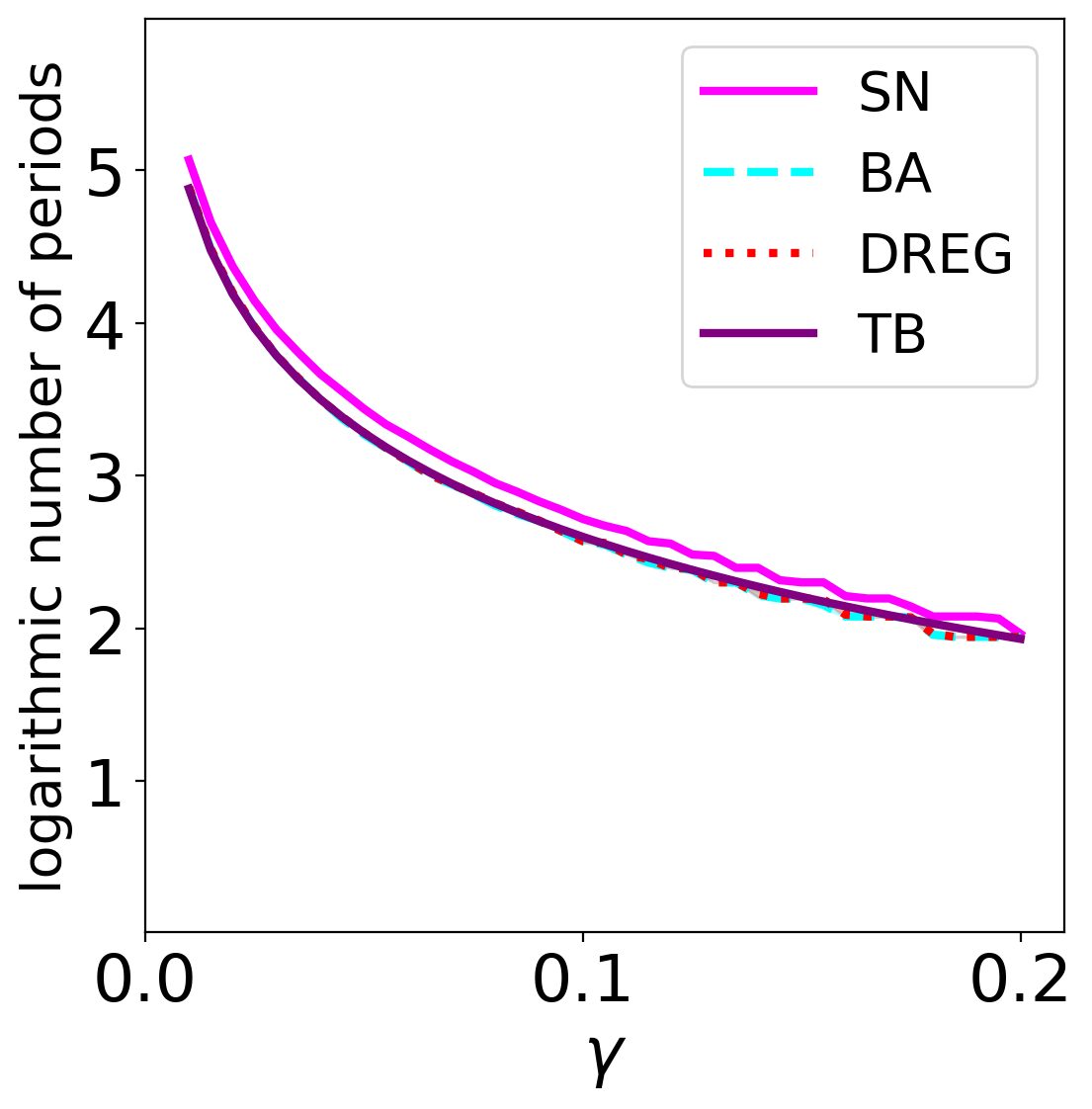}
        \caption{$\alpha=1, \beta=0.025$}
        \label{fig:EX2FB}
    \end{subfigure}
    \begin{subfigure}{0.18\textwidth}
    \captionsetup{font=scriptsize}
        \includegraphics[scale=0.22]{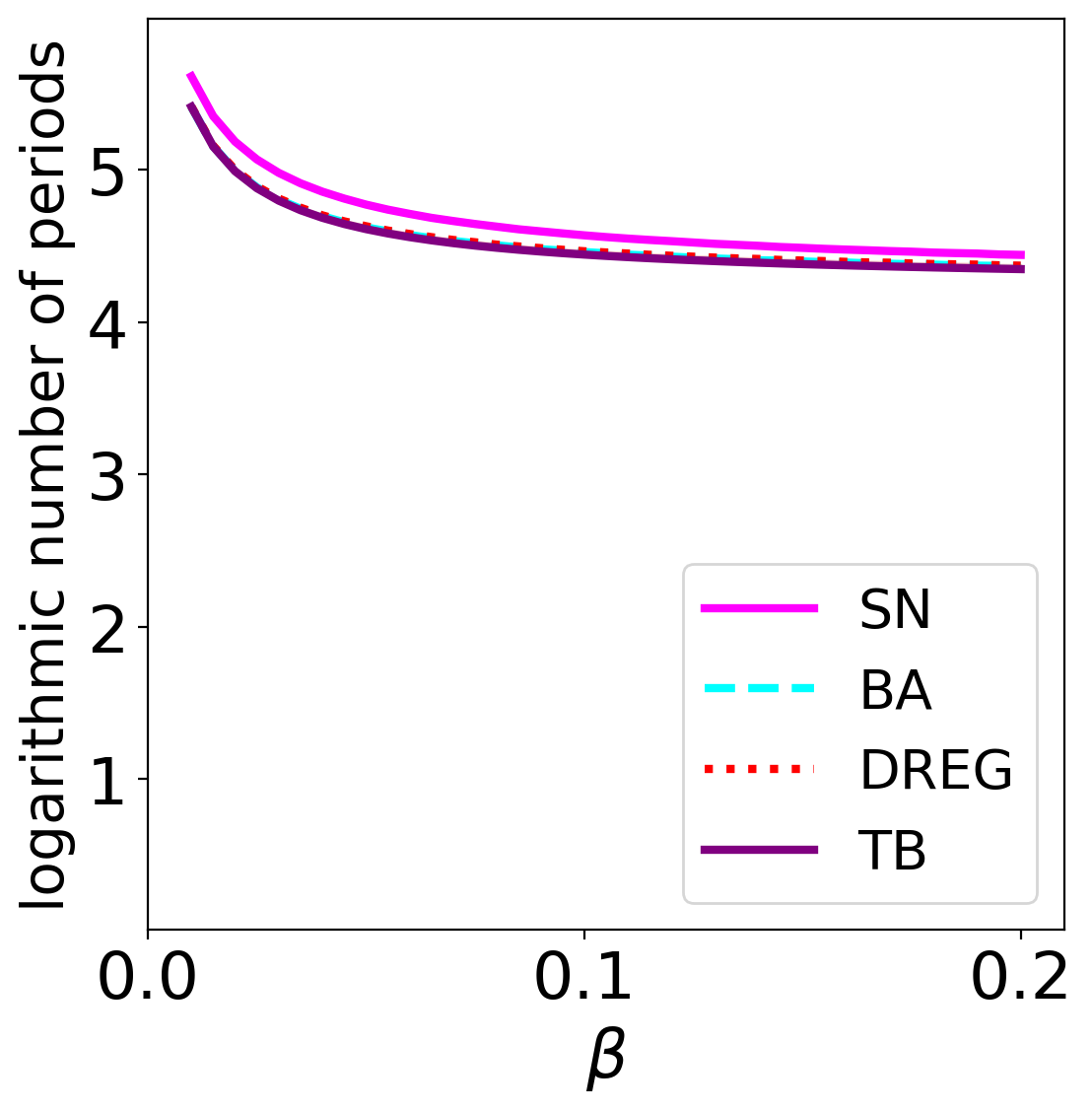}
        \caption{$\alpha=1, \gamma=0.01$}
        \label{fig:EX3FB1}
    \end{subfigure}
    \begin{subfigure}{0.18\textwidth}
    \captionsetup{font=scriptsize}
        \includegraphics[scale=0.22]{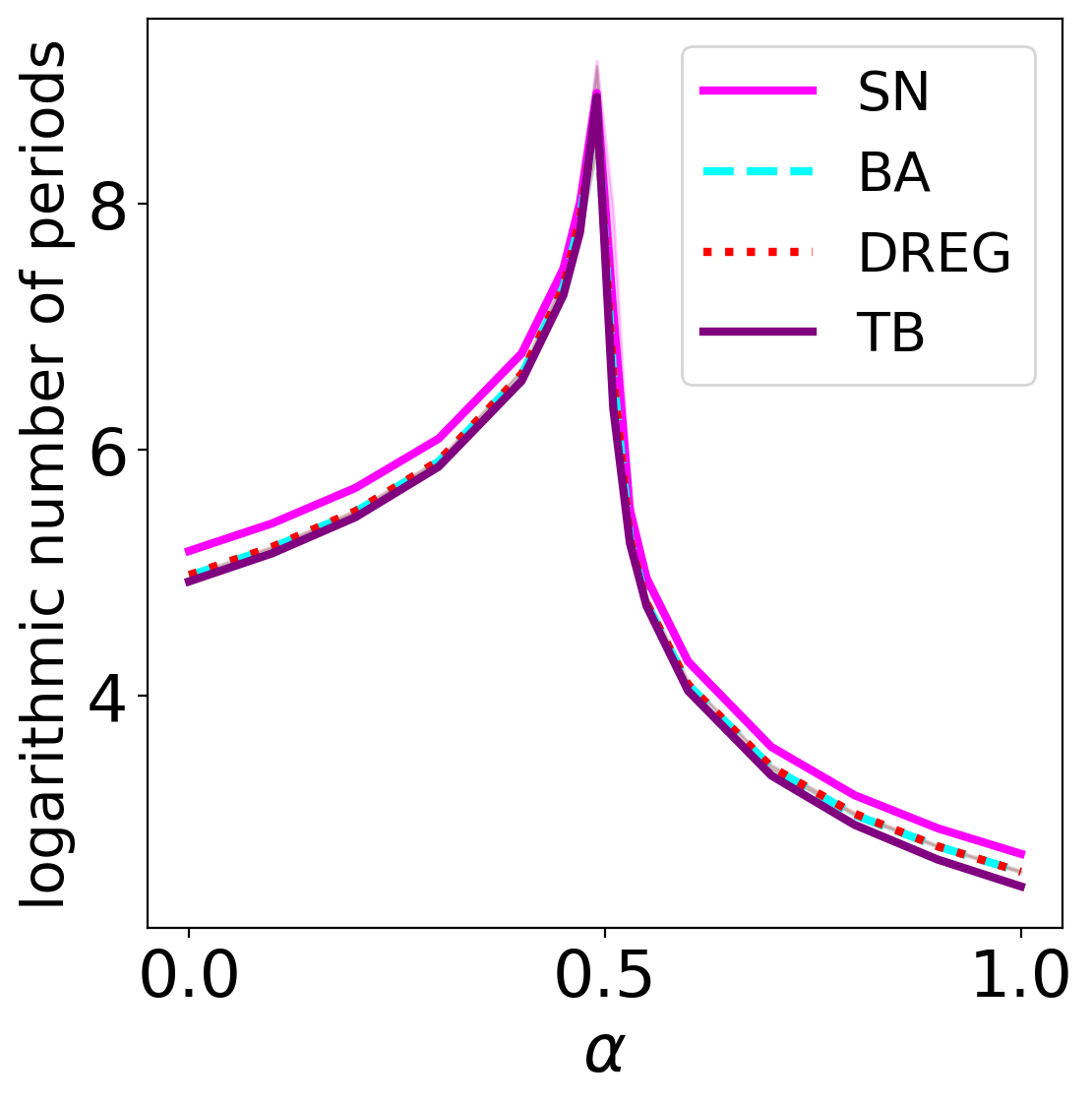}
        \caption{$\gamma=0.1$, $\beta=0.025$}
        \label{fig:EX4FB_T}
    \end{subfigure}
    \begin{subfigure}{0.18\textwidth}
    \captionsetup{font=scriptsize}
        \includegraphics[scale=0.22]{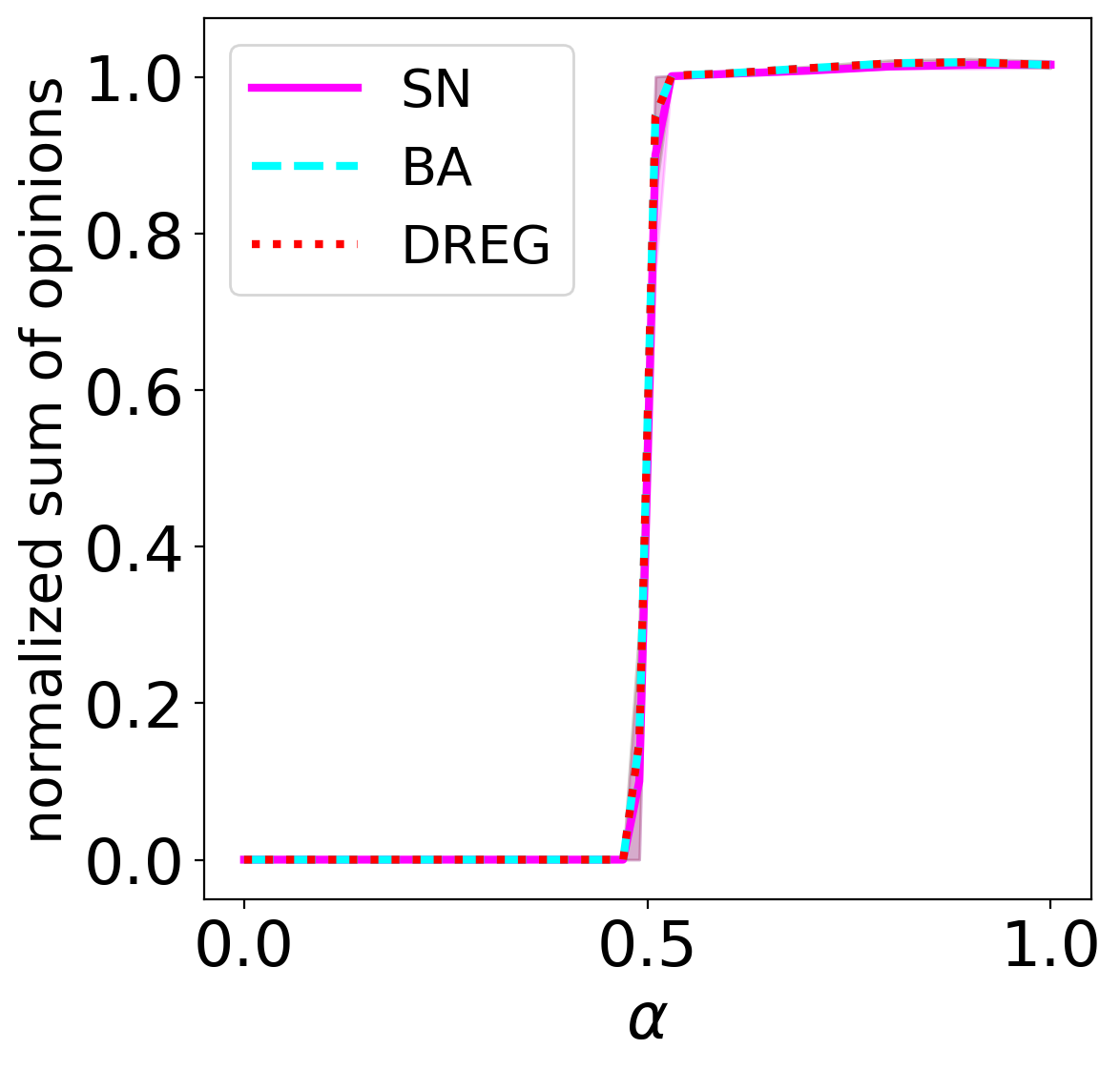}
        \caption{$\gamma=0.1, \beta=0.025$}
        \label{fig:EX4FB_norm}
    \end{subfigure}\\
    \begin{subfigure}{0.18\textwidth}
    \captionsetup{font=scriptsize}
        \includegraphics[scale=0.22]{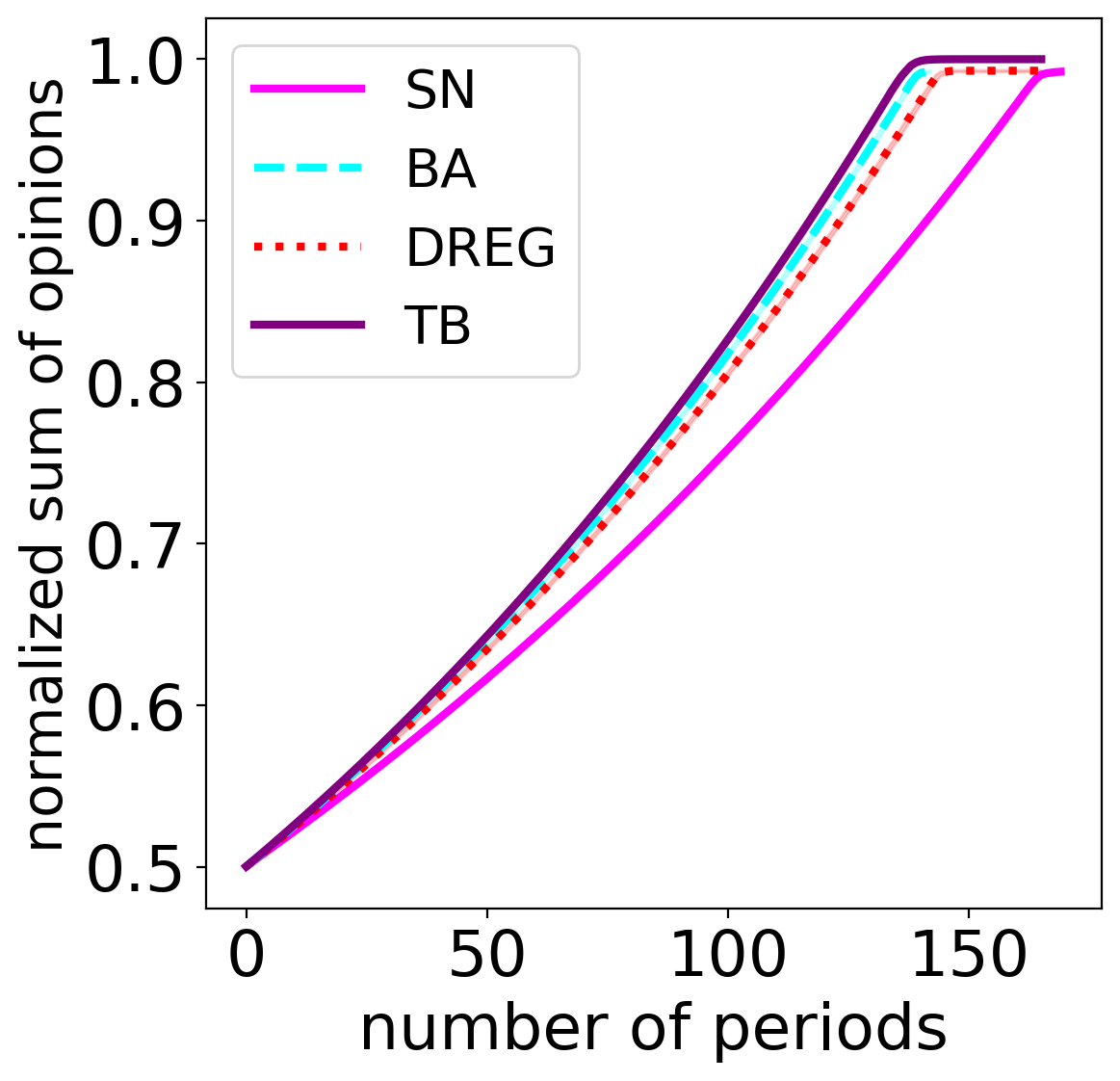}
        \caption{$\alpha = 1, \gamma=0.01, \beta=0.035$}
        \label{fig:EX1WK}
    \end{subfigure}
    \begin{subfigure}{0.18\textwidth}
    \captionsetup{font=scriptsize}
        \includegraphics[scale=0.22]{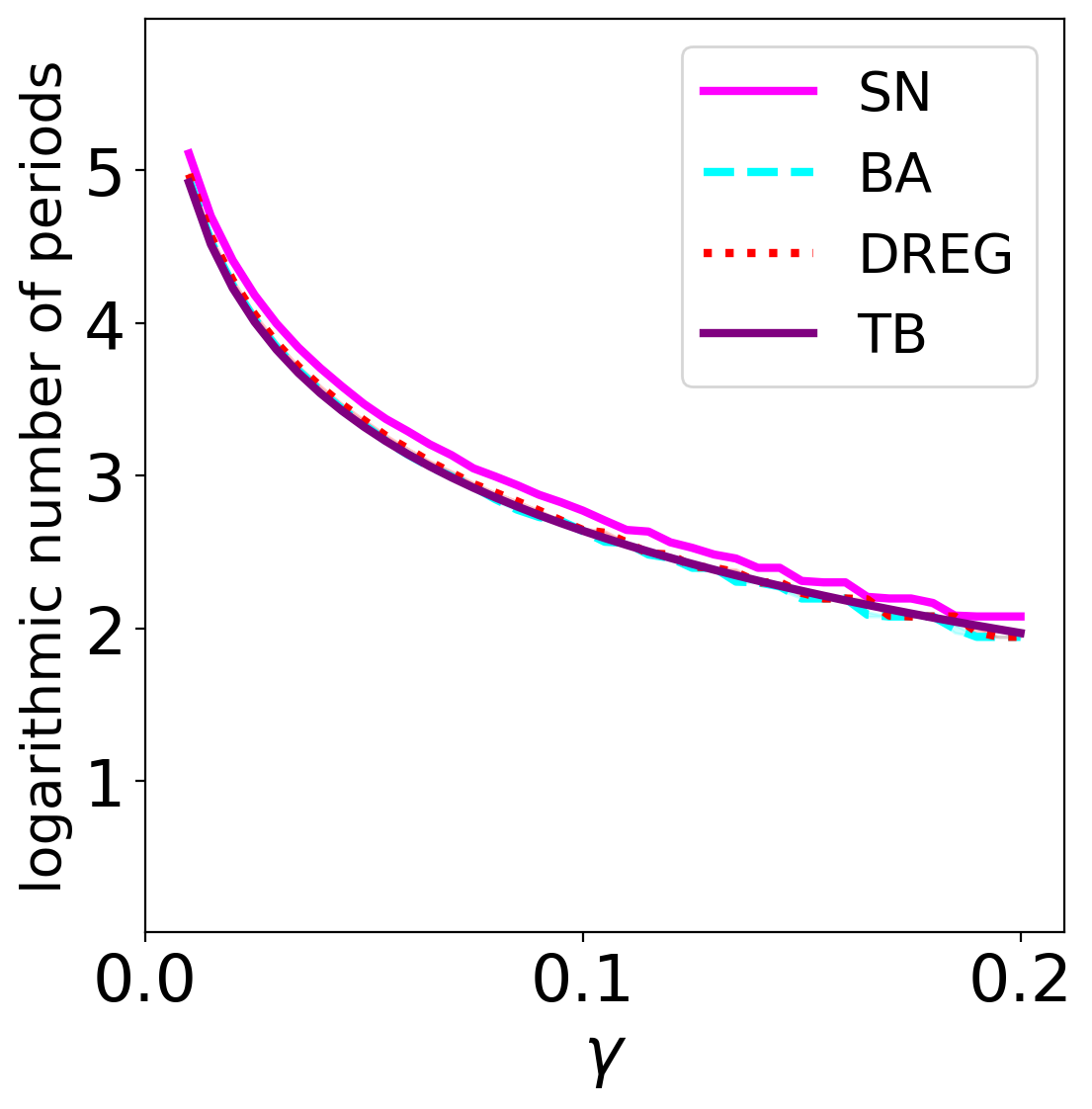}
        \caption{$\alpha = 1, \beta=0.0355$}
        \label{fig:EX2WK}
    \end{subfigure}
    \begin{subfigure}{0.18\textwidth}
    \captionsetup{font=scriptsize}
        \includegraphics[scale=0.22]{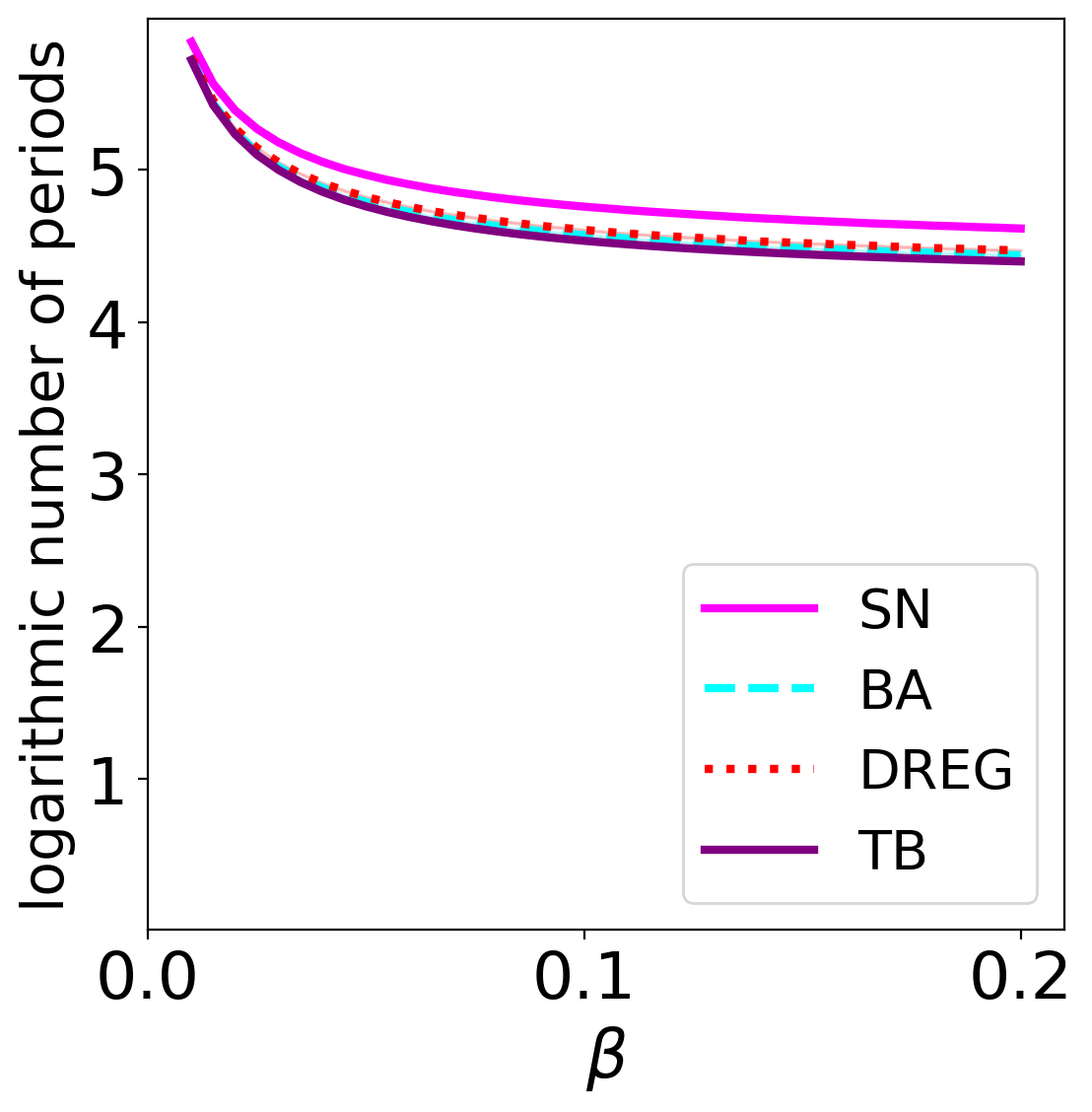}
        \caption{$\alpha=1$,$\gamma=0.01$}
        \label{fig:EX3WK}
    \end{subfigure}
    \begin{subfigure}{0.18\textwidth}
    \captionsetup{font=scriptsize}
        \includegraphics[scale=0.22]{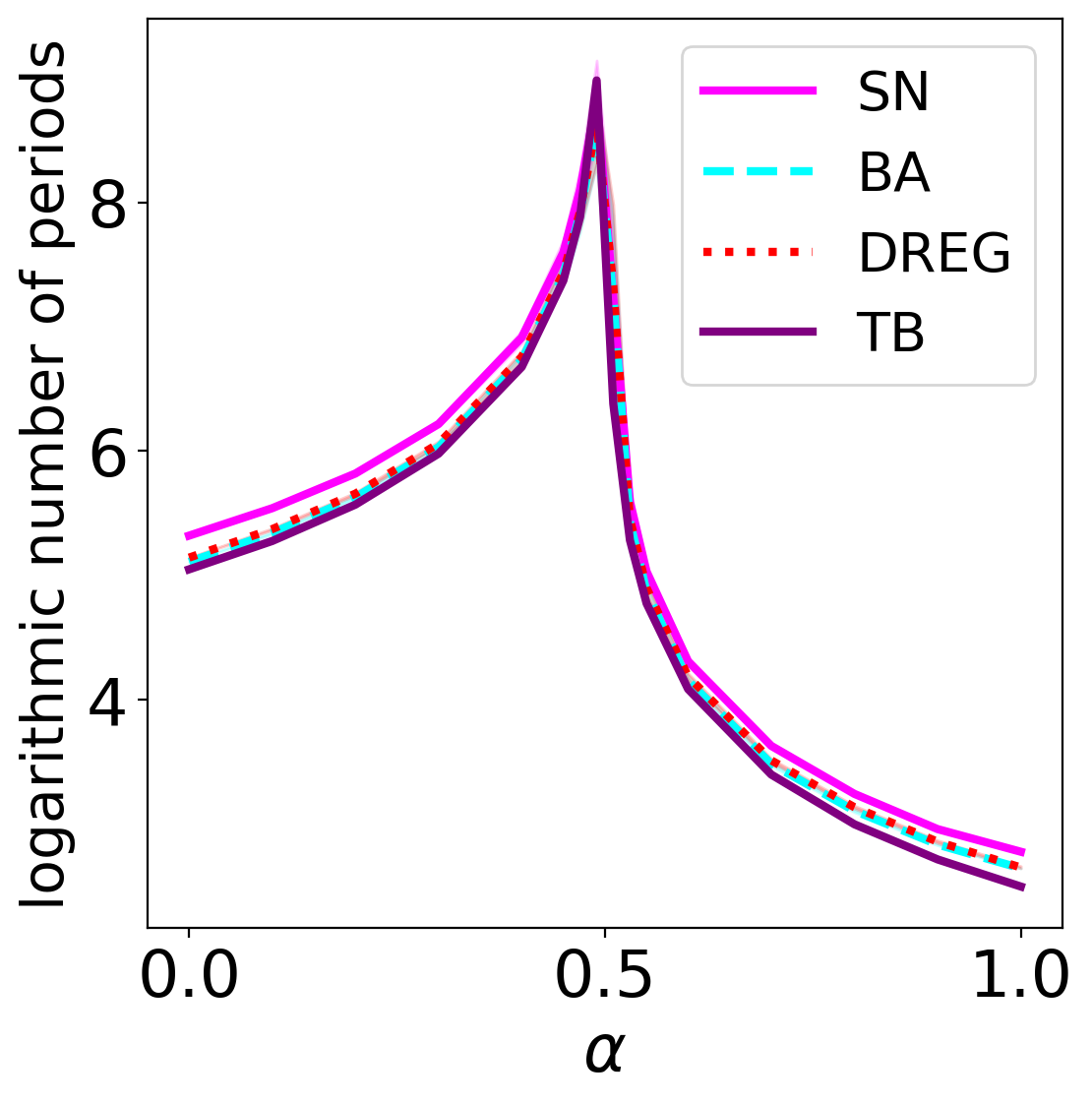}
        \caption{$\gamma=0.1, \beta=0.035$}
        \label{fig:EX4WK_T}
    \end{subfigure}
    \begin{subfigure}{0.18\textwidth}
    \captionsetup{font=scriptsize}
        \includegraphics[scale=0.22]{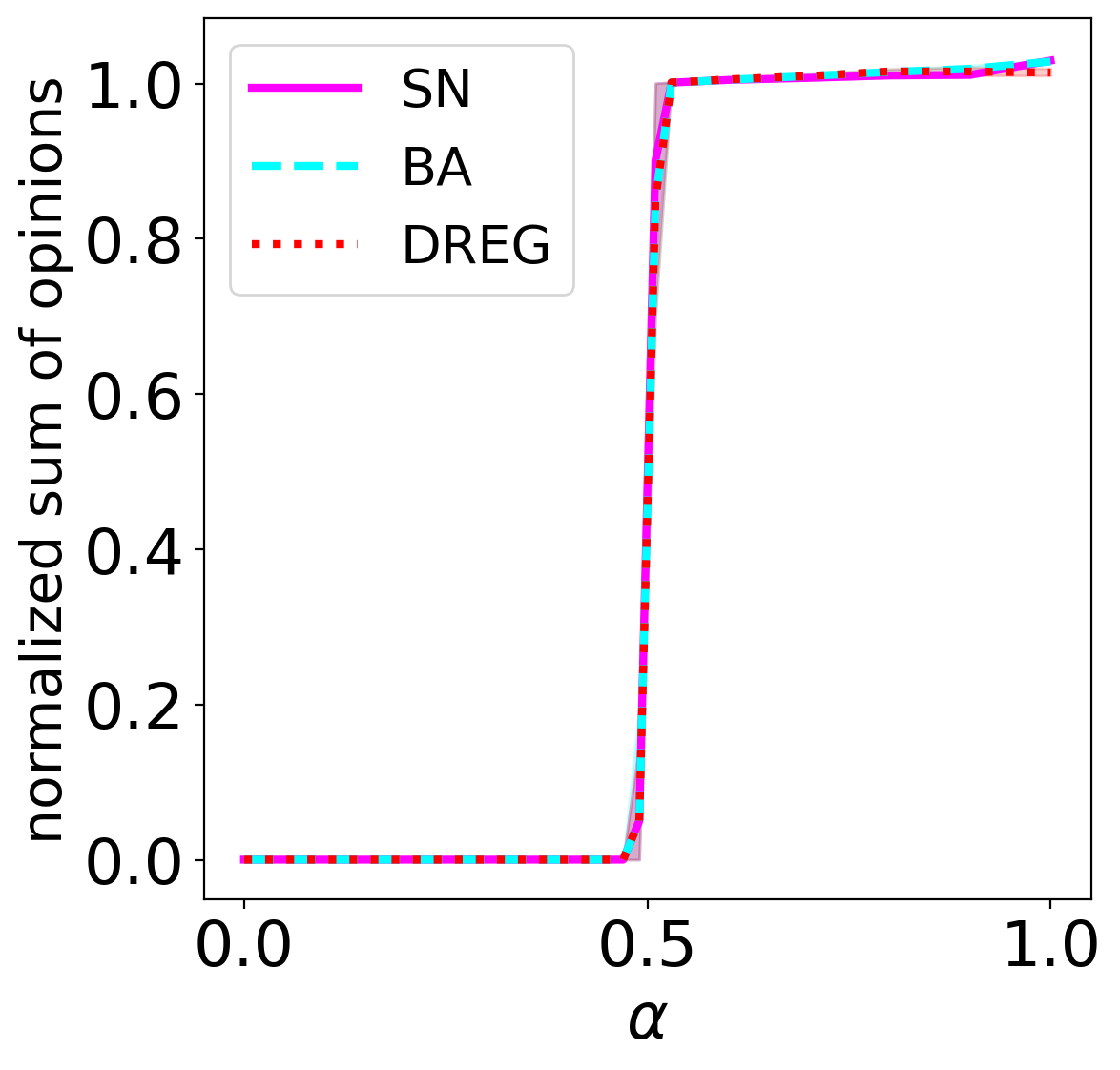}
        \caption{$\gamma=0.1, \beta=0.035$}
        \label{fig:EX4WK_norm}
    \end{subfigure}
    \caption{
    In $(a), (b), (c), (d)$ and $(e)$ we consider the FB SN and the BA and DREG graph with comparable parameters to FB. In $(f),(g),(h),(i)$ and $(j)$ we consider the WK SN and the BA and DREG graph with comparable parameters to WK. In $(a)$ and $(f)$ the normalized sum of expressed opinions over multiple periods is depicted. The logarithmic number of periods to reach normalized average opinion $1/(1 + \gamma)$ is depicted in $(b)$ and $(g)$ for different values of $\gamma$ and in $(c)$ and $(h)$ for different values of $\beta$. In $(d)$ and $(i)$ the logarithmic number of periods to reach normalized average opinion $1/(1+\gamma)$ or $\epsilon := 10/n$ for different values of $\alpha$ are shown. Lastly, in $(e)$ and $(j)$ the final normalized sum of opinions at the end of the process for different values of $\alpha$ are given. The graph TB depicts the theoretical bound from \cref{cor:twomediaoneroundreg} in $(a)$ and $(f)$ and \cref{lem:numroundsphase1multimedia} for $(b),(g),(c),(h)$ and $(d)$ and $(i)$.}
    \label{fig:experimentalresults}
\end{figure*}

Next, we run experiments to validate how well our theoretical bounds match the behavior in on real-world graph data and synthetic graph models.

\subsection{Setup}

\para{Real-world Networks.} For our experiments, we use publicly available Social Network (SN) data from \cite{snapnets}. Our experiments were conducted on Facebook SN ($4039$ nodes and $88234$ edges) and Wikipedia SN ($7115$ nodes and $103689$ edges) datasets; we abbreviate them as FB and WK, respectively.

\para{Synthetic Graphs.} We also conducted experiments on synthetic graphs, namely Barabási-Albert (BA) graphs and $d$-regular random graphs (DREG), i.e., random graphs with a uniform distribution over all $d$-regular graphs on $n$ nodes. We include the BA graph as it is a model to simulate real-world SNs and the DREG since we provide theoretical results, particularly for regular graphs. The parameters in these synthetic graphs were chosen such that they are comparable to the real-world SNs, i.e., such that the (expected) number of nodes/edges is the same as in the aforementioned real-world networks. For example, the DREG graph comparable to the FB SN has $4039$ nodes and degree $44 \approx (2\times 88234)/4039$. All edges in these networks are of weight $1$, except from edges connected to the media source (edge $(i,M)$ has weight $\beta(1+d_i)$ similarly as in the theoretical setup).
 
\para{Innate Opinions.} In our experiments, the innate opinions are chosen from a Gaussian distribution with mean $0.5$ and variance $0.2$. We choose the mean in this way such that on average the innate opinions are not biased towards $0$ or $1$. There is nothing unique about our choice for the variance, and our results would hold for different values of the variance as well. 

\para{Theoretical Bound.} When reporting our results, we sometimes include a plot corresponding to one of our theoretical results; we denote this plot by Theoretical Bound $(TB)$.

 \para{Implementation.} To compute $\md{\finop}^*$ as in \cref{eq:multiplemediazstar}, we rely on the algorithm of \cite{gao2023robust} and its implementation in \hyperlink{https://github.com/danspielman/Laplacians.jl}{Laplacians.jl}. To generate the synthetic graphs, we rely on the implementations in \cite{Graphs2021}. Furthermore, our experiments are implemented in Julia and our code is available in the supplementary material.

 \para{Repetitions. }Each experiment is repeated $20$ times. In the plots in \cref{fig:experimentalresults} the average output with confidence intervals are depicted (note that the repetitions are highly concentrated). In the plots in \cref{fig:experimentalresultsalphahalf}, each one of the $20$ repetitions is plotted individually.

\vspace{-1em}
\subsection{Findings}

\para{Results for a Single Media Source.} In line with \cref{thm:twomediaoneround}, \cref{fig:EX1FB} and \cref{fig:EX1WK} show that when all nodes are connected to a single source $M$ $(\alpha=1)$ the normalized sum of opinions converges to a value arbitrarily close to $1$. Interestingly, the BA graph also follows the theoretical bound for DREG (depicted in purple) quite closely in both networks. \cref{fig:EX2FB}, \cref{fig:EX2WK} and \cref{fig:EX3FB1} \cref{fig:EX3WK} indicate that both bias parameter $\beta$ and external influence magnitude $\gamma$ are negatively correlated with the number of periods to reach normalized average opinion $1/(1 +\gamma)$ (see also \cref{lem:numroundsphase1multimedia}). However, the dependence on $\beta$ appears to be less significant, as we previously observed in \cref{cor:twomediaoneroundreg} for the regular graph.

\begin{figure*}[h]
    \centering
    \begin{subfigure}{0.23\textwidth}
     \captionsetup{font=scriptsize}
        \includegraphics[scale=0.25]{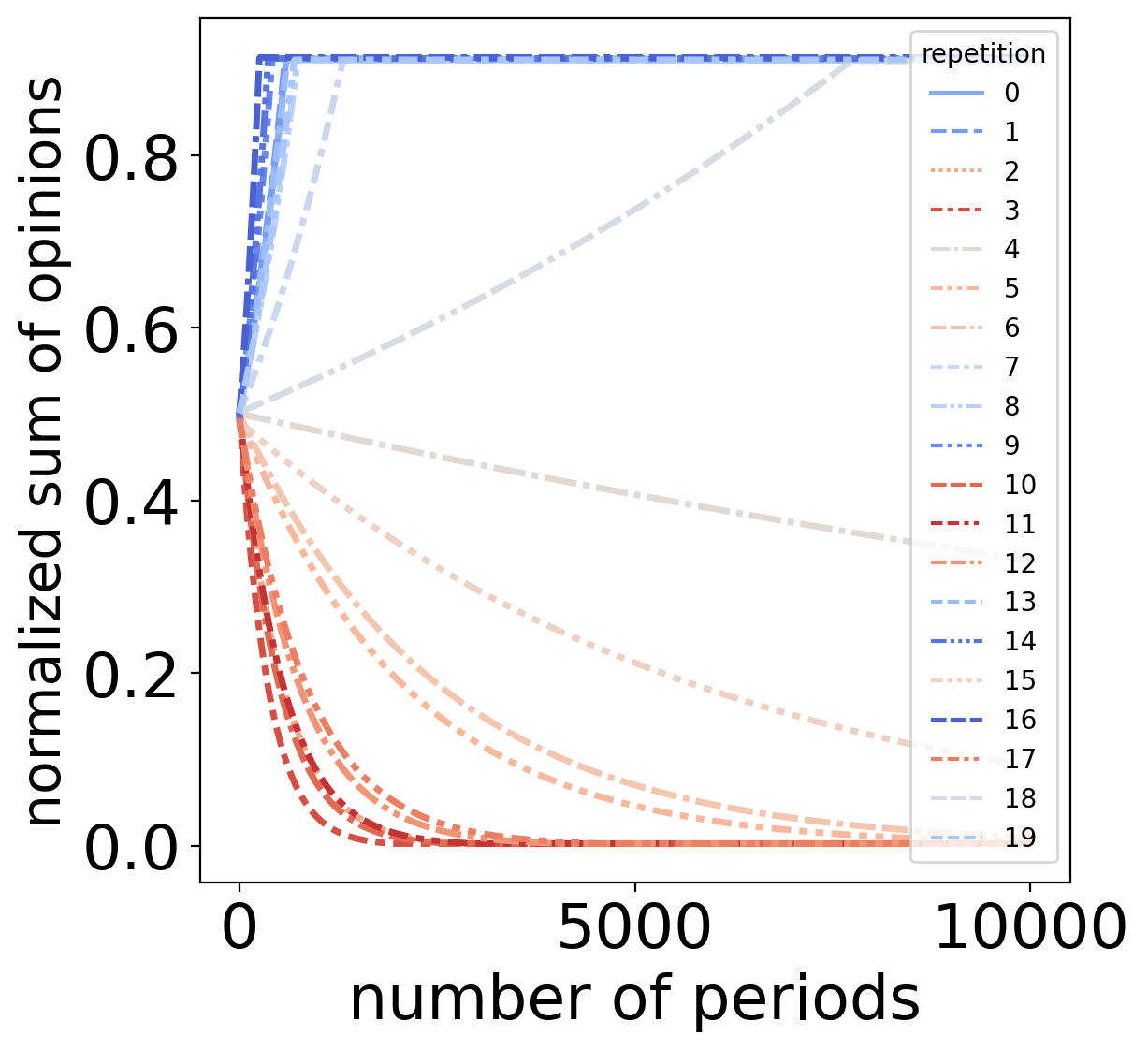}
\caption{$\alpha=1,\gamma=0.01,\beta=0.035$}
        \label{fig:sumNCNE_FB_SN_even}
    \end{subfigure}
    \begin{subfigure}{0.23\textwidth}
    \captionsetup{font=scriptsize}
        \includegraphics[scale=0.25]{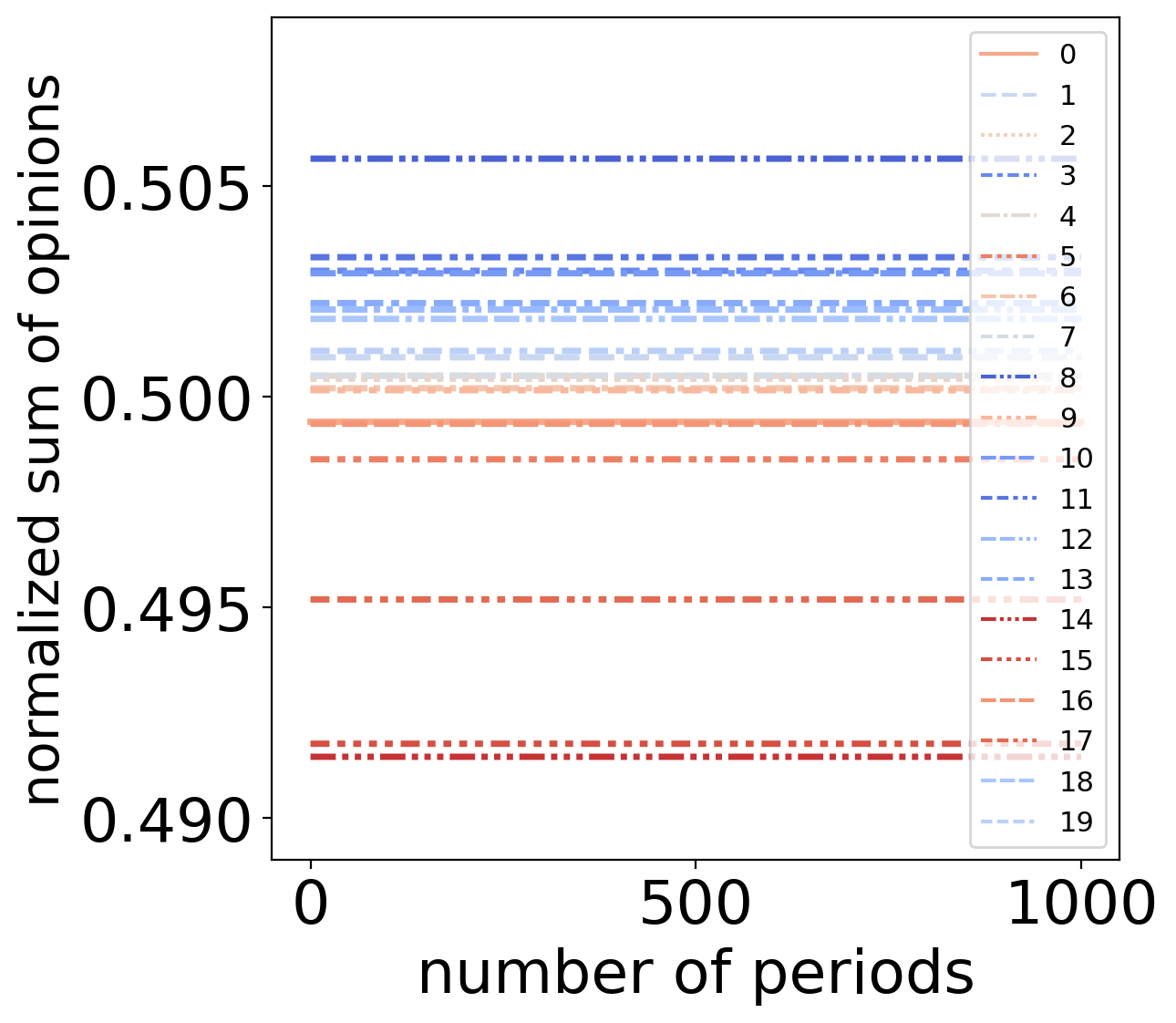}
        \caption{$\alpha=1, \beta=0.035$}
        \label{fig:sumNCNE_FB_DREG_even}
    \end{subfigure}
    \begin{subfigure}{0.23\textwidth}
    \captionsetup{font=scriptsize}
        \includegraphics[scale=0.25]{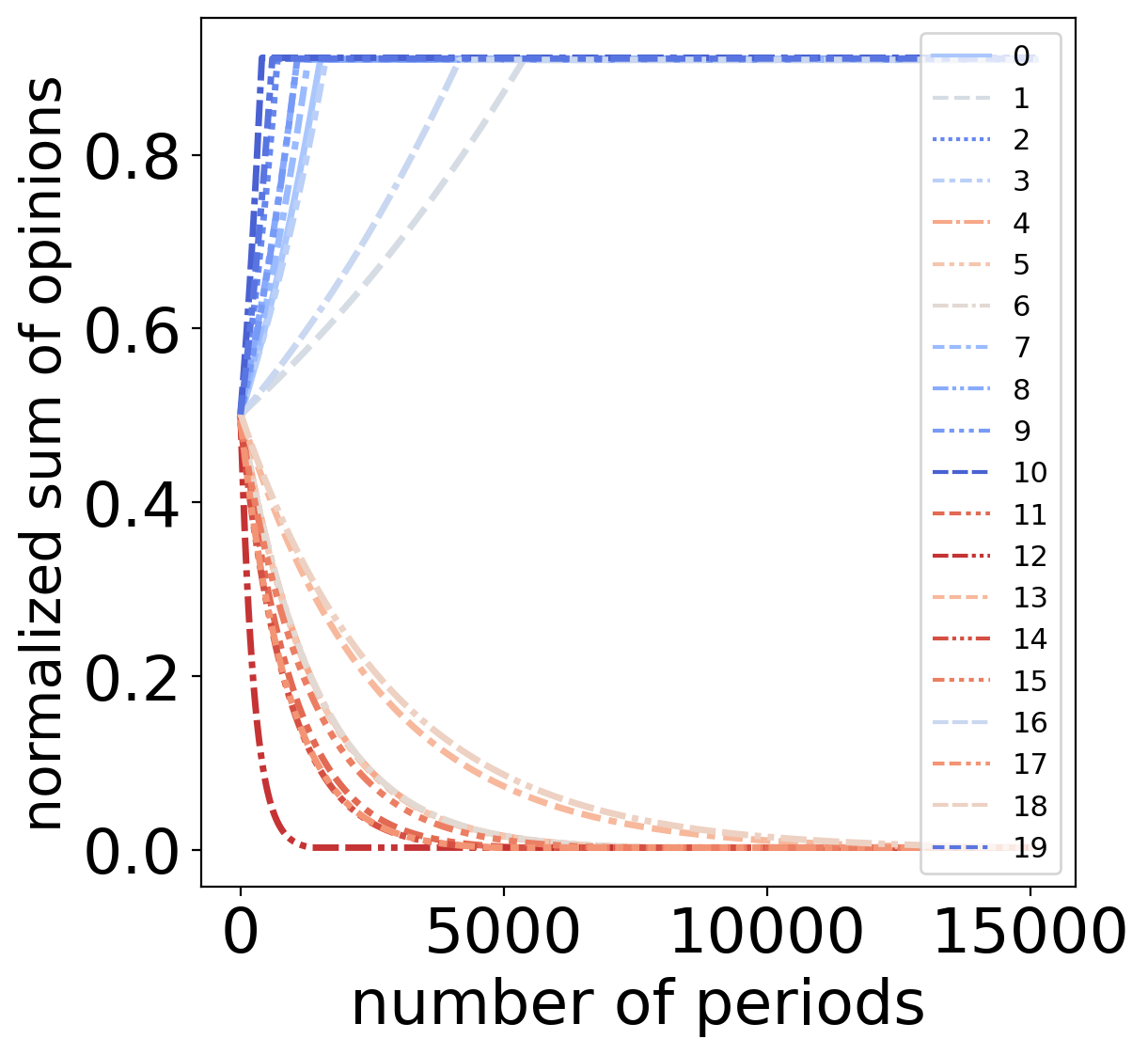}
        \caption{$\alpha=1, \gamma=0.01$}
        \label{fig:sumNCNE_FB_SN}
    \end{subfigure}
    \begin{subfigure}{0.23\textwidth}
    \captionsetup{font=scriptsize}
        \includegraphics[scale=0.25]{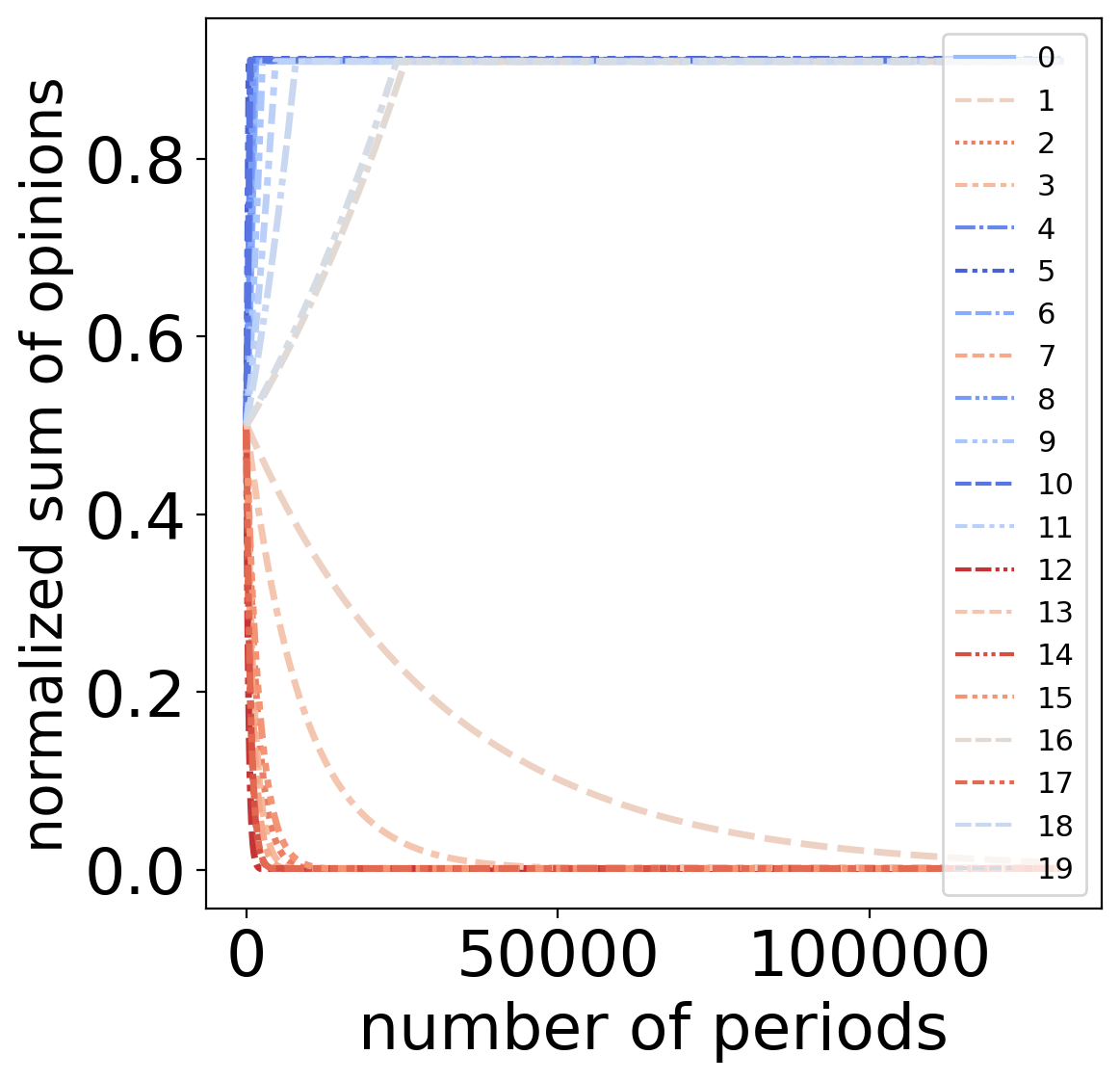}
        \caption{$\gamma=0.1$, $\beta=0.035$}
        \label{fig:sumNCNE_FB_DREG}
    \end{subfigure}
    \begin{subfigure}{0.23\textwidth}
    \captionsetup{font=scriptsize}
        \includegraphics[scale=0.25]{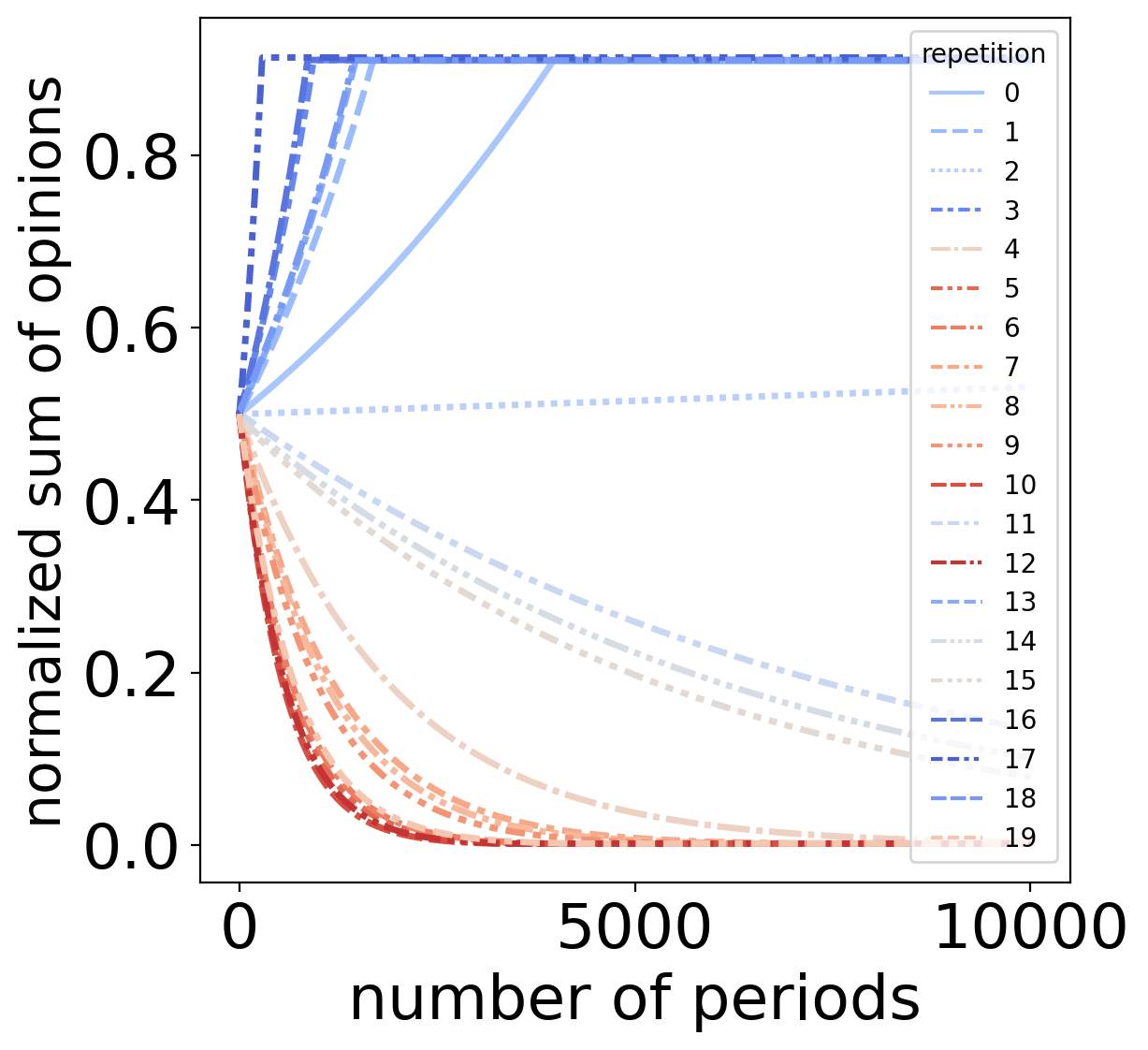}
        \caption{$\alpha = 0.5, \gamma=0.1, \beta=0.5$}
        \label{fig:sumNCNE_WK_SN_even}
    \end{subfigure}
    \begin{subfigure}{0.23\textwidth}
    \captionsetup{font=scriptsize}
        \includegraphics[scale=0.25]{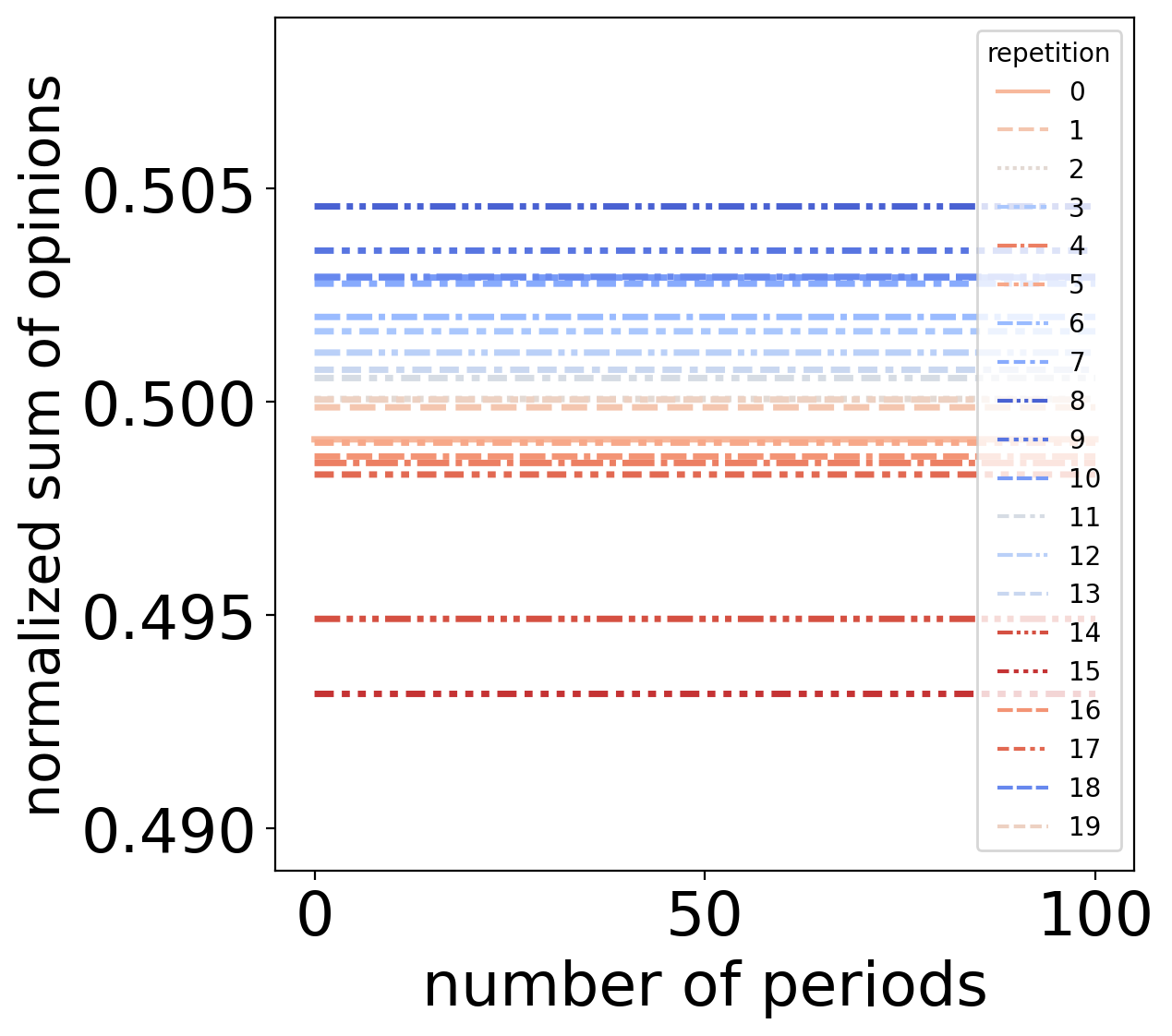}
        \caption{$\alpha = 0.5, \gamma=0.1, \beta=0.5$}
        \label{fig:sumNCNE_WK_DREG_even}
    \end{subfigure}
    \begin{subfigure}{0.23\textwidth}
    \captionsetup{font=scriptsize}
        \includegraphics[scale=0.25]{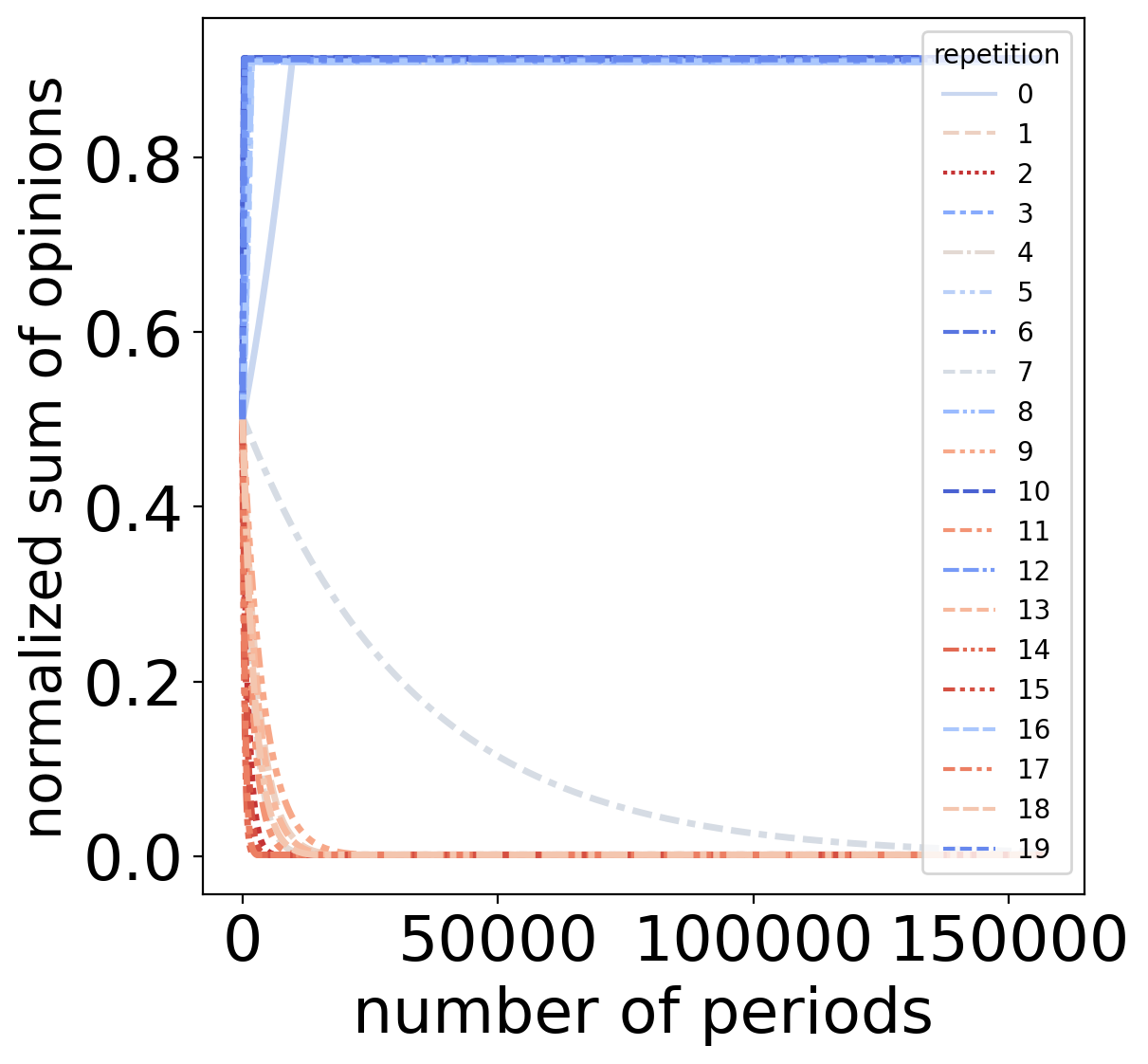}
        \caption{$\alpha \approx 0.5, \gamma=0.1, \beta=0.5$}
        \label{fig:sumNCNE_WK_SN}
    \end{subfigure}
    \begin{subfigure}{0.23\textwidth}
    \captionsetup{font=scriptsize}
        \includegraphics[scale=0.25]{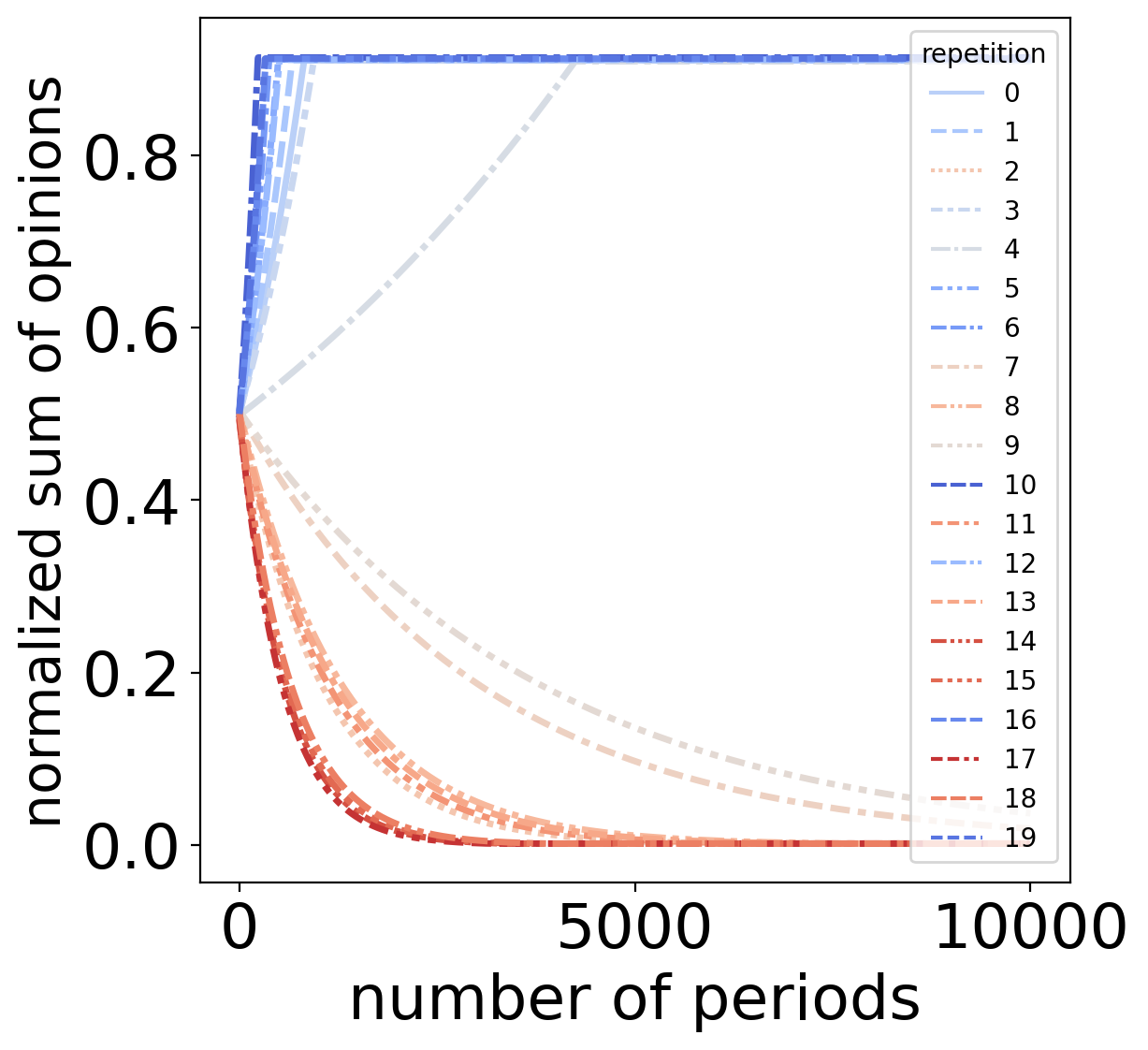}
        \caption{$\alpha \approx 0.5, \gamma=0.1, \beta=0.5$}
        \label{fig:sumNCNE_WK_DREG}
    \end{subfigure}
    \caption{
     In all the above figures, we plot the normalized sum of opinions when $\alpha =0.5$ or $\alpha\approx 0.5$. Each one of the $20$ repetitions is plotted individually. In $(a)$ and $(e)$ we consider the FB and WK SN respectively with a randomly chosen node removed. Figures $(b)$ and $(f)$ correspond to DREG graphs with $n=4038$, $d=44$ and $n=7114$, $d=30$ respectively (note that these are comparable parameters to the FB and TW SN with a single node removed, which is the same across repetitions). In $(c)$ and $(g)$ we depict the FB and WK SN respectively. Lastly, in $(d)$ and $(h)$ DREG graphs with comparable parameters to FB and WK SN respectively are shown.}
    \label{fig:experimentalresultsalphahalf}
\end{figure*}

\para{Results for Multiple Media Sources.} In \cref{fig:EX4FB_T} and \cref{fig:EX4WK_T} the logarithm of the number of periods it takes to reach normalized average opinion $1/(1 +\gamma)$ for different values of $\alpha$ is depicted. Consistent with \cref{lem:numroundsphase1multimedia}, we observe that the process takes the longest for values of $\alpha$ close to $0.5$ in both networks. The asymmetry in the plot is due to the nature of our model (from $n/2$, it takes less time to reach $n/(1+\gamma)$ by increasing with a factor of $(1+\gamma)$ than to reach a very small constant by decreasing by a factor of $1-\gamma$). In \cref{fig:EX4FB_norm} and \cref{fig:EX4WK_norm} we observe that the normalized sum of opinions converges for values of $\alpha < 0.5$ to $0$ and for $\alpha > 0.5$ to $1$, as previously proven for the regular graphs in \cref{cor:twomediaoneroundreg}.

Next, we further investigate the behavior at the threshold value $\alpha=0.5$. As the number of nodes in FB and WK SN are $4039$ and $7115$ respectively, and thus odd (so $\alpha$ cannot be exactly $0.5$ in this case), we delete one node from these graphs uniformly at random to achieve $\alpha=0.5$. Moreover, we generate DREG graphs on $4038$ nodes with $d=44$ (the same average degree as FB SN) and $7114$ and $d=30$ (same average degree as WK SN).

The normalized sum of opinions over multiple periods for these graphs are depicted in \cref{fig:sumNCNE_FB_SN_even}, \cref{fig:sumNCNE_WK_SN_even} (SN's with a randomly chosen deleted node) and \cref{fig:sumNCNE_FB_DREG_even}, \cref{fig:sumNCNE_WK_DREG_even} (DREG graphs with comparable parameters to FB and WK SN with a single removed node). We observe that in line with \cref{lemma:twomediaoneroundalphahalf}, the sum of expressed opinions stays the same in the DREG graphs on $4038$ and $7114$ nodes but on the FB and WK SN converges to either $0$ (this happens in iterations whose graphs are colored red in \cref{fig:sumNCNE_FB_SN_even}, \cref{fig:sumNCNE_WK_SN_even}) or to $1$ (this happens in iterations whose graphs are colored blue in \cref{fig:sumNCNE_FB_SN_even} and \cref{fig:sumNCNE_WK_SN_even}). Lastly, we perform the same experiment on the actual FB and WK SN and DREG graphs with comparable parameters to FB and WK SN as depicted in \cref{fig:sumNCNE_FB_SN} and \cref{fig:sumNCNE_WK_SN} and \cref{fig:sumNCNE_FB_DREG}, \cref{fig:sumNCNE_WK_DREG}. Here, due to the odd number of nodes of FB and WK SN, either $\alpha = 0.4999$ or $\alpha=0.5001$, both with probability $0.5$. We thus observe in \cref{fig:sumNCNE_FB_DREG} that by adding a single node to the DREG graph the final set of opinions also becomes radicalized, even though it takes many periods. This indicates that even a very small difference in the power of the external media source has a significant impact on the final opinion configuration.

\para{Summary.} Our experiments have shown that our simulations results and the bounds that we obtained theoretically match very well. This is interesting since some of our theoretical results were mostly derived for $d$-regular graphs, and the real-world networks are not regular. This empirically indicates that our theoretical results transfer to real-world graphs.

\vspace{-0.5em}
\section{Conclusions}
In this paper, our goal was to obtain a mathematical understanding of how external sources impact the opinion formation process in social networks. To study this formally, we proposed a generalized version of the popular FJ model and derived analytic bounds on the power of the external sources. Several of our bounds are tight for regular graphs, and we showed experimentally that our theoretical bounds closely match simulation results on real-world datasets.

In the future, it would be interesting to study the more general setup, where an external source can connect to only a given number of nodes and aims to optimize a specific objective, such as maximizing/minimizing the final sum of opinions. Another potential avenue for future research is investigating the impact of external sources on polarization in the network, especially when two external sources attempt to pull the opinions in opposite directions.

\vspace{-0.5em}
\begin{acks}
This research has been funded by the Vienna Science and Technology Fund (WWTF) [Grant ID: 10.47379/VRG23013],
the ERC Advanced Grant REBOUND (834862), and the Wallenberg AI, Autonomous Systems and Software Program (WASP) funded by the Knut and Alice Wallenberg Foundation.
\end{acks}

\bibliographystyle{ACM-Reference-Format}
\balance
\bibliography{sample-base}

\clearpage

\appendix

\end{document}